\documentclass{article}
\usepackage[ruled,vlined]{algorithm2e} 

\usepackage[margin=1in]{geometry}

\usepackage{bbding}
\usepackage[nointegrals]{wasysym}
\usepackage[utf8]{inputenc}
\usepackage{xspace}
\usepackage{mathtools}
\usepackage{url}
\usepackage{subcaption}
\usepackage{tikz}
\usetikzlibrary{arrows, fit, positioning}
\usetikzlibrary{backgrounds,automata,calc}
\usetikzlibrary{decorations.pathreplacing}
\usepackage{csquotes}
\usepackage{amsmath,amsfonts,amsthm}
\usepackage{bm,bbm}
\usepackage{pgfplots}

\usepackage{enumerate}
\usepackage{paralist}
\usepackage[shortlabels]{enumitem}
\usepackage{appendix}
\usepackage[capitalize]{cleveref}
\usepackage{thmtools}
\usepackage{mathtools}
\usepackage[numbers]{natbib}
\usepackage{changepage}
\usepackage{authblk}
\usepackage{commath}

\newcommand{\USW}{\texttt{USW}\xspace}
\newcommand{\NSW}{\texttt{NSW}\xspace}

\newcommand{\R}{\mathbb{R}}
\newcommand{\Z}{\mathbb{Z}}

\renewcommand{\cal}[1]{\mathcal{#1}}

\usepackage{mathtools}

\setlist{topsep=0.5ex,itemsep=0.1ex}

\newtheorem{theorem}{Theorem}[section]
\newtheorem{lemma}[theorem]{Lemma}

\newtheorem{claim}[theorem]{Claim}

\newtheorem*{theorem*}{Theorem}

\theoremstyle{definition}

\newtheorem{definition}[theorem]{Definition}
\newtheorem{obs}[theorem]{Observation}
\newtheorem{example}[theorem]{Example}

\usepackage{color-edits}

\addauthor{AE}{magenta}
\addauthor{VV}{blue}
\addauthor{YZ}{purple}

\title{Stable Matching under Matroid Rank Valuations}
\author[1]{Alon Eden}
\author[2]{Vignesh Viswanathan}
\author[2]{Yair Zick}
\affil[1]{Hebrew University of Jerusalem, Israel}
\affil[2]{University of Massachusetts, Amherst, USA}
\date{}

\begin{document}

\maketitle

\begin{abstract}
We study a two-sided matching model where one side of the market (hospitals) has combinatorial preferences over the other side (doctors). Specifically, we consider the setting where hospitals have matroid rank valuations over the doctors, and doctors have either ordinal or cardinal unit-demand valuations over the hospitals. While this setting has been extensively studied in the context of one-sided markets, it remains unexplored in the context of two-sided markets.

When doctors have ordinal preferences over hospitals, we present simple sequential allocation algorithms that guarantee stability, strategyproofness for doctors, and approximate strategyproofness for hospitals. When doctors have cardinal utilities over hospitals, we present an algorithm that finds a stable allocation maximizing doctor welfare; subject to that, we show how one can maximize either the hospital utilitarian or hospital Nash welfare. Moreover, we show that it is NP-hard to compute stable allocations that approximately maximize hospital Nash welfare.  
\end{abstract}

\section{Introduction}

Stable matching is a fundamental problem in the EconCS community. 
In the classical version, we have a set of $n$ doctors and $n$ hospitals, each with ordinal preferences over the other side.
The goal is to find a stable matching of doctors to hospitals; informally, a matching is stable if no doctor-hospital pair prefers each other to their assigned match. 
This problem has been studied since the 1950s and is well understood today. 
There exists a polynomial time algorithm (the Gale-Shapley algorithm) \citep{galeshapley1962og} that outputs a stable matching. Moreover, the set of stable matchings have a lattice structure \cite{knuth1997stablemarriage}; as a consequence, there exists a stable matching that provides all doctors their best possible outcome among all the stable matchings. 

Moving beyond the classical problem, several natural generalizations have been studied in the literature. 
One line of generalizations has been the many-to-one matching problem where a hospital can be matched with (or allocated to) multiple doctors; this is the model used by National Residency Matching Program \citep{roth1999residentmatching}. 
\citet{galeshapley1962og} propose a model motivated by college admissions where each hospital has a cardinality constraint upper bounding the number of doctors they can be matched with. However, if we assume that each hospital recruits several doctors, it is reasonable to assume that hospitals would want residents trained in a diverse set of specialties; therefore, typically, hospitals have preferences over \emph{sets} of doctors. Indeed, recent work builds on this model by adding structure 
to the acceptable sets of doctors for each hospital. In the CS community, \citet{fleiner2001matroidmatching} studies a model where each hospital has a matroid constraint dictating the set of doctors they can be allocated. In the economics community, \citet{hatfield2005matching} present a similar constraint based on the assumption that doctors are substitutes. 
For all these models, Gale-Shapley-style algorithms have been shown to output stable allocations. Moreover, these problems admit a similar lattice structure over the set of stable allocations, under the assumption that the preferences of each side over the other forms a total order; that is, there are no ties in the preferences. 

A second line of generalization considers ordinal preferences with ties and incompleteness. In most cases, a stable matching can be trivially computed by breaking ties arbitrarily. The goal of this line of research is instead to compute a stable matching that maximizes the matching size \citep{mcdermid2009stable,csaji2023matroidmatching,paluch2014approxmatching}. This problem is surprisingly NP-hard even in the classical one-to-one matching setting \citep{halldorson2007approxmatching,manlove2002approxmatching,iwama1999smti}. More generally, when preferences have ties and can be incomplete, stable matchings are no longer guaranteed to have a lattice structure. 

In this paper, we present and study a natural problem at the intersection of these two lines of research. Following the model of \citet{fleiner2001matroidmatching}, we assume there is a matroid constraint for each hospital that their allocated set of doctors must satisfy. 
However, we assume that there is no preference order that hospitals have over the doctors, and that the goal of each hospital is simply to maximize the number of doctors they are allocated subject to the matroid constraint. 
In more technical terms, we assume hospitals have {\em matroid rank valuations}. 
This class of preferences has recently gained popularity in the one-sided markets literature due to its expressivity and the fact that the matroid structure can be leveraged to design fair, efficient and strategyproof mechanisms \cite{babaioff2021EF,barman2021matroid,benabbou2020finding,viswanathan2022generalyankee,berczi2021matroidpricing}. \textit{This paper is the first to extensively study matroid rank valuations in two-sided markets.}

Mathematically, the class of matroid rank valuations is simple enough that the existing impossibility results for hospital strategyproofness do not apply~\citep{roth1985hospitalstrategyproofness}. 
However, this class of preferences does not admit the lattice structure that previous models exploited to understand stable matchings. This makes it an interesting class of preferences to study with many open questions. 
We focus on the design of \emph{strategyproof} mechanisms that output \emph{stable matchings} with high \emph{market efficiency}. 
We study two variants of doctor utilities, resulting in two fundamental research questions.
\begin{enumerate}[(Q1)]
    \item 
    Can we design strategyproof mechanisms that output stable matchings with high market efficiency, when doctors have ordinal preferences? \label{item:q1}
    \item 
    Can we compute stable allocations that maximize various welfare objectives 
    for either side of the market (or both simultaneously), when doctors have cardinal preferences? \label{item:q2}
\end{enumerate}
Both questions have been studied in most many-to-one matching models. 
\citet{roth1985hospitalstrategyproofness} shows that it is impossible to guarantee hospital strategyproofness when hospitals have ordinal preferences over the doctors along with a cardinality constraint. 
\citet{hatfield2005matching} present a doctor-strategyproof algorithm for a general class of preferences which outputs the optimal stable allocation for all the doctors\footnote{In the many-to-one matching model, when there is a lattice structure over the set of stable matchings, the optimal stable allocation for the set of doctors trivially maximizes the size of the matching.}; given the earlier impossibility result, this is the best result possible.

When the set of stable matchings has a lattice structure, the best outcome for one side is well-defined and can be computed by a variant of the Gale-Shapley algorithm (see \citep{hatfield2005matching} and \citep{fleiner2001matroidmatching} for examples). 
Moreover, this is the worst outcome for the other side. 
However, with no clear lattice structure, as in our model, it is unclear what objectives can be maximized subject to stability. 
Addressing these questions with respect to matroid rank valuations offers useful insight into stable allocations in the absence of a clear lattice structure, as well as provides tools to design mechanisms in these settings. 

\subsection{Our Contributions and Techniques}
\paragraph{Addressing \ref{item:q1}}
We present two mechanisms. 
Our first mechanism, called the high welfare serial dictatorship (HWSD), outputs a stable allocation that maximizes hospital utilitarian welfare and is strategyproof for the doctors. Moreover, when hospitals have binary OXS valuations, a subclass of matroid rank valuations \cite{paesleme2017gross}, the mechanism is  $2$-approximately strategyproof for the hospitals  (\Cref{thm:red-jacket}). 
Note that in our setting, hospital utilitarian welfare is equal to matching size, a notion of efficiency studied in prior work \citep{mcdermid2009stable,csaji2023matroidmatching}.

HWSD works as follows: doctors take turns picking a hospital. 
Each doctor picks the highest-ranked hospital that is willing to accept them, i.e., accepting a doctor will not violate the hospital's matroid constraint; however, this match is allowed only if there exists a \emph{completion} of the resulting partial allocation that maximizes hospital utilitarian welfare.    

Stability, doctor strategyproofness and max hospital utilitarian welfare follow straightforwardly from the description of the algorithm. 
The main non-trivial result we show for HWSD is the hospital strategyproofness guarantee. 
To show approximate hospital strategyproofness, we characterize the optimal misreport that any hospital can make. We use this to show that no hospital can gain significantly by misreporting. This broad proof structure has been used to show exact strategyproofness in both fair division \citep{babaioff2021EF} and many-to-one matching \citep{hatfield2005matching}. However, to show approximate strategyproofness, we analyze allocations via {\em reversible path augmentations} (Lemmas \ref{lem:fTreport} and \ref{lem:manip-half}). 
Reversible path augmentations build on the technique of path augmentations used in prior work \citep{babaioff2021EF} to show strategyproofness with matroid rank valuations in the one-sided setting. 
The key difference in our proof is that we need to take doctor preferences into account in the two-sided setting; this makes our analysis quite intricate. 

We also show that the same analysis can be used to prove \textit{exact hospital strategyproofness} of the HWSD mechanism when hospitals have \textit{binary capped additive valuations}. 

The second mechanism we study is the popular serial dictatorship (SD) mechanism. Here, doctors may pick any hospital that is willing to admit them, regardless of whether the resulting partial assignment can be completed to a utilitarian welfare optimal assignment. 
SD is strategyproof for doctors, $2$-approximately strategyproof for hospitals \textit{for general matroid rank valuations}, and $2$-approximately maximizes hospital utilitarian welfare (\Cref{thm:serial-dictatorship}). 
Here again, the most non-trivial result is showing approximate hospital strategyproofness. 
To prove this, we use the simplicity of the serial dictatorship mechanism to upper bound the change in the doctors hospitals are assigned due to misreports. 
The analysis involves a careful accounting of the set of hospitals who were denied a match with doctors due to some hospital misreporting, and the set of doctors whose allocation changed due to misreporting.
\paragraph{Addressing \ref{item:q2}.} We focus on both the utilitarian and the Nash welfare objectives for both hospitals and doctors. 
The HWSD mechanism described above efficiently finds a stable outcome that maximizes hospital (unconstrained) utilitarian welfare. 
We complement this result by showing, rather unfortunately, that finding a stable allocation that maximizes Nash welfare for hospitals is NP-hard for any finite approximation factor (\Cref{thm:nash_hardness}). 

Moving to cardinal doctor utilities, 
we present an efficient algorithm that outputs a stable allocation maximizing doctor utilitarian welfare (\Cref{lem:max-doctor-welfare-max-hospital-USW}) or doctor Nash welfare. 
Moreover, we show that subject to maximizing doctor utilitarian welfare (or doctor Nash welfare), we can efficiently maximize hospital Nash welfare (\Cref{thm:max-nash-welfare}). 
This is the most technically involved result of the paper. 

Maximizing doctor utilitarian welfare can be reduced to the problem of finding a max weight independent set at the intersection of two matroids. 
Maximizing hospital Nash welfare subject to maximizing doctor utilitarian welfare can be done via a local search algorithm. Most algorithms for matroid intersection use path augmentations in the {\em matroid exchange graph} \citep{brezovec1986matroidintersection,Chakrabarty2019MatroidIntersection}. 
The proof for our local search algorithm's correctness also follows from an analysis of the matroid exchange graph. Specifically, we show that we can find a set of cycles in the exchange graph where at least one such cycle is an augmenting path with some desirable properties. 
We then show that repeatedly finding these desirable cycles and augmenting along them eventually increases the hospital Nash welfare without reducing doctor utilitarian welfare. 
This analysis generalizes existing results in the fair division literature\footnote{Specifically, we generalize the computational result of \citet{babaioff2021EF} who show that a max Nash welfare allocation can be computed in the one-sided setting when agents have matroid rank valuations.}, and may be of independent interest.

\vspace{0.1in}

\paragraph{Main open question.} The main open question left by our work is \textit{``Does there exist a hospital-strategyproof mechanism that outputs a stable allocation when hospitals have matroid rank valuations?''}

\subsection{Additional Related Work}
Variants of the many-to-one matching problem have been studied in the literature, both by the computer science and the economics community. 
In the economics literature, the college admissions problem \citep{galeshapley1962og} was generalized to the job matching problem \citep{knoer1981jobmatching,kelso2982gs,roth1984gametheory}, and later generalized to the stable contract matching problem \citep{hatfield2005matching}. 
Both the job matching and the contract matching problem use some natural definitions of substitutes to constrain hospital preferences. 
For all three problems, Gale-Shapley style algorithms compute a stable matching. 
The college admissions and the job matching problems have also been generalized to the many-to-many matching setting, with similar results on the existence and computability of stable allocations \citep{roth1984polarizing,roth1985jobmatching}. 
Existing results on hospital strategyproofness (or lack thereof) and manipulability of the college admissions problem \citep{roth1985hospitalstrategyproofness,sonmez2004manytoone,sonmez1997manipulation} carry over to all the models described above.

The computer science community studies similar generalizations. 
The key difference being that the class of preferences used are described by mathematical objects (like matroids) rather than economic principles (like substitutability). 
\citet{fleiner2001matroidmatching} generalize the classical stable matching problem to the \emph{matroid kernel problem}.  The matroid kernel problem is a many-to-many matching problem where each hospital (resp. doctor) has ordinal preferences and a matroid constraint over the set of doctors (resp. hospitals). 
A Gale-Shapley style algorithm outputs stable allocations for this problem as well. 
This model has been further generalized to handle lower quotas on each hospital's bundle \citep{huang2010classified,yokoi2017polymatroid,fleiner2016lowerquota}, as well as ties in the preferences \citep{csaji2023matroidmatching,kamiyama2022superstable}. 

The problem of matching with ties in preferences has also received considerable attention. \citet{irving1994indifference} introduces the problem and provides definitions for stability in the presence of ties. 
\citet{manlove2002indifferencelattice} and \citet{roth1984gametheory} present examples where one-to-one matching instances do not have a hospital-optimal (weakly) stable matching; 
this provides evidence of the loss of structure in the problem when ties are introduced. \citet{manlove2002approxmatching} show that generalizing beyond ties to incomplete preferences with ties renders many previously easy problems NP-hard, even in the one-to-one matching setting.  
A specific problem of interest is finding the maximum size stable matching, with the best known approximation ratio being $1.5$ \citep{paluch2014approxmatching,halldorson2007approxmatching,mcdermid2009stable}.

Our problem can be seen as a special case of the \emph{matroid kernel with ties problem}. 
Two papers study the matroid kernel with ties  problem. 
\citet{csaji2023matroidmatching} present a $1.5$-approximation algorithm for the maximum size stable matching problem in this model, generalizing the $1.5$-approximation of \citet{mcdermid2009stable}, and show a matching lower bound. 
\citet{kamiyama2022superstable} studies the existence and computation of super-stable allocations in this model. Allocations are super-stable if moving any doctor to some hospital does not weakly improve both agents' outcomes; that is, it strictly worsens the outcome of one of the agents involved. 
Our focus is instead on weak stability (as defined by \citet{irving1994indifference}). Since hospitals do not have a preference order over doctors in our model, we believe this to be a more meaningful stability notion.

In recent independent work, \citet{aziz2025strategyproofmaximummatchingdichotomous} consider the many-one matching setting, and show that it is possible to compute a stable matching which maximizes the matching size and is strategyproof for all agents involved {\em when hospital have binary capped additive valuations}. In Corollary \ref{corr:red-jacket-SP-capped-additive}, we present a similar result. However, we note that their result is a little more general since they show that their algorithm satisfies additional properties like non-bossiness and Pareto optimality.

\section{Preliminaries}\label{sec:preliminaries}
We use $[k]$ to denote the set $\{1, 2, \dots, k\}$. Given a set $S$ and an element $d$, we use $S+d$ and $S-d$ to denote $S \cup \{d\}$ and $S \setminus \{d\}$ respectively. 

We have a set of $n$ {\em hospitals} $H = [n]$, and a set of $m$ {\em doctors} $D = \{d_1, \dots, d_m\}$. Each hospital $h\in H$ has a {\em matroid rank valuation} function $v_h: 2^D \rightarrow \Z$ over the set of doctors; for any set of doctors $S \subseteq D$, $v_h(S)$ denotes the value of the set of doctors $S\subseteq D$ according to the hospital $h$. We use $\Delta_h(S, d) = v_h(S+d) - v_h(S)$ to denote the \emph{marginal gain} of adding the doctor $d$ to the set $S$ according to  hospital $h$. A valuation function $v_h$ is a matroid rank function (MRF) if it satisfies the following three properties:
\begin{inparaenum}[(a)]
    \item $v_h(\emptyset) = 0$, 
    \item for each $S \subseteq D$, $d \in D \setminus S$, we have $\Delta_h(S, d) \in \{0, 1\}$, and 
    \item for each $S \subseteq T \subseteq D$, $d \in D \setminus T$, $\Delta_h(S, d) \ge \Delta_h(T, d)$.
\end{inparaenum}
MRFs have the following useful property. 

\begin{obs}[Matroid augmentation property~\cite{oxley2011matroid}] \label{obs:augmentation}
Let $v_h$ be a matroid rank function over a set $D$ of elements. For every $S,T\subseteq D$ such that $v_h(S)<v_h(T)$, there exists an element $d \in T \setminus S$ such that $\Delta_h(S, d) = 1$.
\end{obs}

We assume every doctor $d\in D$ has a strict and complete preference order $\succ_d$ over the set of hospitals $H$; we write $h_1 \succ_d h_2$ if the doctor $d$ prefers the hospital $h_1$ to $h_2$. We define cardinal preferences for doctors in Section \ref{sec:cardinal-utilities}. 

An {\em allocation}  $X = (X_0; X_1, \dots, X_n)$ is an $(n+1)$-partition of the set of doctors $D$. Each hospital $h$ is allocated $X_h$ and $X_0$ denotes the set of unallocated doctors. We use $X(d)$ to denote the hospital that doctor $d$ was allocated to in the allocation $X$. We use the term allocation instead of matching to differentiate our setting from the one-to-one matching setting. 
Given an allocation $X$, the {\em utility} of hospital $h$ is the value $v_h(X_h)$.

When analyzing the time complexity of our mechanisms, we assume we can efficiently make {\em value queries} to the hospital valuation functions. 
That is, we assume we can compute the value $v_h(S)$ for any hospital $h$ and set of doctors $S$ in polynomial time. 

As described in the introduction, we assume each hospital's matroid rank valuation function enforces a constraint over the set of doctors the hospital can be allocated. Specifically, we assume each hospital $h$ can only be allocated sets of doctors which correspond to independent sets in the matroid defined by the rank function $v_h$. This is equivalent to the notion of {\em non-redundancy} used in fair division \citep{benabbou2020finding}. 
\begin{definition}[Non-redundancy]
    An allocation $X$ is said to be {\em non-redundant} if for all $h \in H$, $v_h(X_h) = |X_h|$. A bundle $S \subseteq D$ is said to be non-redundant with respect to $v_h$ (for some hospital $h$) if $v_h(S) = |S|$.
\end{definition}

The main desideratum we require of our allocations is \emph{stability}. 
To define stability, we first define a \emph{blocking pair}. 
We use the same definition of blocking pair as \citet{hatfield2005matching} and \citet{csaji2023matroidmatching}.
\begin{definition}[Blocking Pair]\label{def:blocking-pair}
    Given an allocation $X$, a pair consisting of a set of doctors $S$ and hospital $h$ (denoted $(h, S)$) form a {\em blocking pair} if:
    \begin{enumerate}[(a)]
        \item $S$ is non-redundant with respect of $v_h$ and $v_h(S) > v_h(X_h)$, and 
        \item for all $d \in S$, $h \succeq_{d} X(d)$ with equality holding only if $h = X(d)$.
    \end{enumerate}
\end{definition}
In other words, the hospital $h$ \textit{strictly} prefers the set of doctors $S$ to their assigned bundle $X_h$, and the doctors in $S$ all \textit{weakly} prefer $h$ to their assigned hospital under $X$. 
This definition naturally extends to preference orders with ties as well. When doctors' preference orders have ties in them, the above definition requires the set of doctors in $S$ not already allocated to $h$ to strictly prefer $h$ to their current allocation $X(h)$.
\begin{definition}[Stability]\label{def:stability}
    An allocation $X$ is {\em stable} if it is non-redundant and has no blocking pair. 
\end{definition}
When agents have matroid rank functions, this condition can be simplified. 
We repeatedly use the following observation when proving stability guarantees (proof in Appendix \ref{apdx:prelims}).
\begin{restatable}{obs}{obsblocking} \label{obs:blocking}
Given a non-redundant allocation $X$, a blocking pair exists if and only if there exists some doctor $d$ and hospital $h$ such that $\Delta_h(X_h, d) = 1$ and $h \succ_d X(d)$.
\end{restatable}
We are interested in stable allocations with welfare guarantees. 
An allocation $X$ maximizes hospital \emph{utilitarian welfare} (hospital-\USW for short) if it maximizes $\USW(X) = \sum_{h \in H} v_h(X_h)$. 
An allocation $X$ maximizes hospital \emph{Nash welfare} (hospital-\NSW for short) if it maximizes $\NSW(X)=\prod_{h \in H} v_h(X_h)$. 
An allocation $X$ is an $\alpha$ approximation to hospital utilitarian welfare for $\alpha\ge 1$ if  $\USW(X)\ge \max_{X'}\USW(X')/\alpha$.
We can define these welfare functions for doctors similarly, assuming doctors have cardinal utility functions.
When hospitals have matroid rank valuations, every stable allocation approximately maximizes hospital-\USW (proof in Appendix~\ref{apdx:prelims}).
\begin{restatable}{lemma}{lemgeneralwelfare} \label{lem:stable_welfare}
    When hospitals have matroid rank valuations, every stable allocation is a 2-approximation of the maximum hospital-\USW.
\end{restatable}
\emph{Strategyproofness} is a key desideratum in mechanism design.
We assume that agent valuations are private information, and they report them to the mechanism. 
We say that agent $i$ has an \emph{incentive to misreport} their preferences if there exists a valuation $\hat v$ such that if agent $i$ reports $\hat v$ instead of their true valuation $v_i$ the mechanism outputs an outcome that agent $i$ strictly prefers.
A mechanism is strategyproof if no agent has an incentive to misreport their preferences. 

We distinguish between hospital and doctor strategyproofness. 
A mechanism is hospital (resp. doctor) strategyproof (in short, hospital-SP/doctor-SP) if no hospital (resp. doctor) has an incentive to misreport their preferences.

While we present several doctor-SP mechanisms that yield a stable allocation, identifying mechanisms that are hospital-SP and output stable allocations turns out to be a more demanding task. 
We devise mechanisms to compute stable allocations that are only approximately hospital-SP.
More formally, a mechanism is $\alpha$-approximate hospital-SP for $\alpha\ge 1$, if for any hospital $h$, misreporting preferences results in a utility of at most $\alpha$ times the utility received under a truthful report.
This notion is sometimes called approximate incentive compatability \citep{troebst2024hz}.
\subsection{Preference Classes}
We also present some results for subclasses of matroid rank valuations. 
A valuation function $v_h$ is said to be {\em binary OXS} if there is an unweighted bipartite graph $G = (D \cup R, E)$ for some set $R$ such that for every $S \subseteq D$, $v_h(S)$ is equal to the size of the maximum size matching in the subgraph of $G$ induced by the set of vertices $S \cup R$. The set $R$ is sometimes called the set of {\em slots} for valuation function $v_h$. 

A valuation function $v_h$ is said to be {\em binary capped additive} if for some positive integer $b_i$ and all $S \subseteq D$, $v_h(S) = \min \{b_i, \sum_{d \in S} v_h(\{d\})\}$. 
The following hierarchy relates the three valuation classes studied in this paper:
\begin{align*}
    \text{binary capped additive} \subset \text{binary OXS} \subset \text{matroid rank}.
\end{align*}

\subsection{Impossibility for Hospital-\USW Maximizing Hospital-SP mechanisms}\label{sec:no-sp-usw}
The first result we present is a simple impossibility result. We show that no hospital-SP mechanism that outputs a stable allocation can be hospital-\USW maximizing. This example also highlights the fact that there may be no hospital-optimal (or doctor-optimal) stable allocation in our setting. 

\begin{example}\label{ex:impossibility}
The instance consists of two hospitals $\{h_1, h_2\}$ and three doctors $\{d_1, d_2, d_3\}$. The matroid rank valuation functions for the hospitals are given as follows:
\begin{align*}
    & v_{h_1}(S) = |S \cap \{d_2\}| + \min\{|S \cap \{d_1, d_3\}|, 1\}\\
    & v_{h_2}(S) = |S \cap \{d_3\}| + \min\{|S \cap \{d_1, d_2\}|, 1\}.
\end{align*}
The preferences of the doctors are given as follows: $d_1$ and $d_3$ prefer $h_1$ to $h_2$, and $d_2$ prefers $h_2$ to $h_1$.
There are only two hospital-\USW maximizing 
stable allocations in this instance denoted by $X$ and $Y$: 
(see \Cref{fig:matching-before-deviation}):
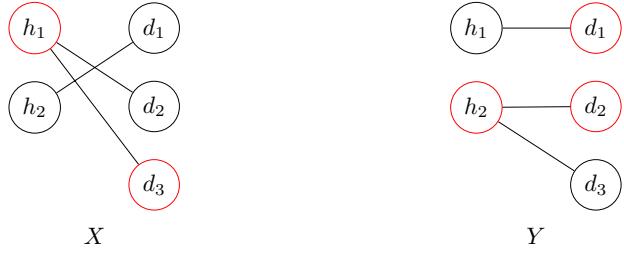
\begin{figure}
    \centering
    \resizebox{0.5\textwidth}{!}{%
    \begin{tikzpicture}[mycirc2/.style={circle,fill=white,minimum size=0.75cm,inner sep = 3pt},
    mycirc3/.style={circle,fill=white,draw=red,minimum size=0.55cm,inner sep = 3pt},
    mycirc/.style={circle,fill=white,draw=black,minimum size=0.55cm,inner sep = 3pt},
    BC/.style = {decorate,  
                     decoration={brace, amplitude=4mm, mirror},
                     thick, pen colour={black}
                    },
    ]
    \node[mycirc3] (l11){$h_1$};
    \node[mycirc, right=1cm of l11] (r11){$d_1$};
    \node[mycirc, below=0.4cm of l11] (l12){$h_2$};
    \node[mycirc, below=0.4cm of r11] (r12){$d_2$};
    \node[mycirc2, below=0.4cm of l12] (l13){};
    \node[mycirc3, below=0.4cm of r12] (r13){$d_3$};
    \path (l13)--(r13) node[midway, below=0.5cm] {$X$};
    \draw (l11)--(r12) (l11)--(r13) (l12)--(r11);
    \node[mycirc, right=4cm of r11] (l21){$h_1$};
    \node[mycirc3, right=1cm of l21] (r21){$d_1$};
    \node[mycirc3, below=0.4cm of l21] (l22){$h_2$};
    \node[mycirc3, below=0.4cm of r21] (r22){$d_2$};
    \node[mycirc2, below=0.4cm of l22] (l23){};
    \node[mycirc, below=0.4cm of r22] (r23){$d_3$};
    \path (l23)--(r23) node[midway, below=0.5cm] {$Y$};
    \draw (l21)--(r21) (l22)--(r22) (l22)--(r23); 
    \end{tikzpicture}
    }
    \caption{The only two stable hospital-welfare maximizing allocations. Red nodes are nodes who receive their best possible outcome, and have no incentive to deviate.}
    \label{fig:matching-before-deviation}
\end{figure}
\begin{align*}
    X_{h_1} &= \{d_2, d_3\} && X_{h_2} = \{d_1\} \\
    Y_{h_1} &= \{d_1\} && Y_{h_2} = \{d_2, d_3\}
\end{align*}
Note that if $h_2$ were to report the valuation $v'_{h_2}(S) = |S \cap \{d_2, d_3\}|$, then the only hospital-\USW maximizing stable allocation is $Y$. 
This is because
\begin{inparaenum}[(a)]
\item any stable allocation must allocate $d_2$ to $h_2$, and
\item $h_1$ accepts $d_1$ or $d_3$ but not both.
\end{inparaenum}
Using a similar argument, we can show that if $h_1$ were to report the valuation function $v'_{h_1}(S) = |S \cap \{d_2, d_3\}|$, then the only hospital-\USW maximizing stable allocation is $X$. 

Let $\cal A$ be a mechanism that always outputs a hospital-\USW maximizing 
stable allocation. On the instance described above, it must either output $X$ or $Y$. 
If it outputs $X$, then $h_2$ can deviate and become strictly better off. If it outputs $Y$, then $h_1$ can deviate and become strictly better off. Therefore, $\cal A$ cannot be strategyproof.
\end{example}

While this example does not completely rule out the possibility of hospital-SP mechanisms, it highlights the difficulty in the problem. The only other stable allocation for the instance above which is not $X$ or $Y$ is $Z$ where $Z_{h_1} = \{d_3\}$ and $Z_{h_2} = \{d_2\}$. In order to guarantee hospital strategyproofness, a mechanism must deliberately output the low hospital-\USW allocation $Z$ either when $h_1$ reports truthfully or when $h_1$ misreports. 
\section{Doctors with Ordinal Preferences}\label{sec:ordinal-utilities}

Our objective is to identify doctor/hospital strategyproof mechanisms that output stable allocations, with additional hospital-\USW guarantees. However, the counterexample in \Cref{sec:no-sp-usw} shows that the ideal version of this result cannot be achieved. 
Instead, we seek mechanisms that achieve these guarantees {\em approximately}. We present two such mechanisms, each providing a different trade-off between hospital-\USW and hospital strategyproofness.

\subsection{High Welfare Serial Dictatorship (HWSD)}

We present a mechanism that outputs a stable allocation, is doctor-SP and  obtains optimal hospital-\USW. The HWSD mechanism operates in a sequential fashion, where at each iteration, we finalize the allocation of one doctor. More specifically, at each iteration, we pick a doctor $d$ and allocate it to its highest ranked hospital $h$ that satisfies the following two conditions:
\begin{inparaenum}[(a)]
    \item the utility of hospital $h$ increases after it receives doctor $d$, and 
    \item the current allocation can be augmented to a hospital-\USW maximizing allocation.
\end{inparaenum}
The mechanism is presented in \Cref{algo:red-jacket}. Missing proofs of this section appear in \cref{apdx:hwsd}.

\begin{algorithm}
    \DontPrintSemicolon
    \caption{High Welfare Serial Dictatorship (HWSD) Mechanism}
    \label{algo:red-jacket}
    $X = (X_0, X_1, \dots, X_n) \gets (D, \emptyset, \dots, \emptyset)$\;
    \For{$i$ in $1$ to $m$}{
        \For{$j$ in $1$ to $n$}{
            Let $h$ be the $j$-th highest ranked hospital for doctor $d_i$\;
            \If{$\Delta_{h}(X_{h}, d_i) = 1$}{
                \If{there exists a Max hospital-\USW allocation $Y$ such that $X_{h'} \subseteq Y_{h'}$ for all $h' \ne h$ and $X_{h} + d \subseteq Y_h$}{
                    $X_{h} \gets X_{h} + d$\;
                    $X_0 \gets X_0 \setminus d$\;
                    Break\;
                }        
            }
        }
    }
    \Return $X$\;
\end{algorithm}

\begin{theorem}\label{thm:red-jacket}
The high welfare Serial Dictatorship (HWSD) mechanism has the following properties:
\begin{enumerate}[(i)]
    \item The mechanism is doctor-SP.
    \item The algorithm runs in polynomial time.
    \item The output allocation maximizes hospital-\USW.
    \item The output allocation is stable.
    \item The mechanism is $2$-approximately hospital-SP when hospitals have binary OXS valuations.
\end{enumerate}
\end{theorem}

Due to space constraints, the proofs for (i)--(iv) are relegated to Appendix \ref{apdx:hwsd}. We only show the hospital-SP guarantees here. Our hospital-SP guarantees only apply when hospital valuations are defined by binary OXS functions. 
This means each hospital's valuations $v_h$ is defined by a bipartite graph $G_h$ with doctors on one side and an arbitrary number of nodes (called {\em slots}) on the other.    
The value of any set of doctors $S \subseteq D$ is the maximum cardinality matching in the subgraph of $G_h$ induced by the doctors in $S$ and all the slots sharing edges with $S$. In other words, the value $v_h(S)$ is the number of slots that can be filled by the set of doctors $S$, where each doctor $d$ can only take up neighboring slots in the graph $G_h$.


\Cref{algo:red-jacket} is $2$-approximately hospital-SP when hospitals have binary OXS valuations. That is, for any hospital, misreporting their valuation can increase a hospital's bundle size by no more than a factor of $2$. 

To prove this property, we begin with a few definitions.
Given a subset of doctors $T \subseteq D$, we define the valuation $f_T$ as $f_T(S) = |S \cap T|$; that is, each doctor in $T$ adds a value of $1$. Given an allocation $X$ and a set of doctors $D'$, we define the allocation $X \setminus D'$ as the allocation where each hospital $h$ receives the set $X_h \setminus D'$.
Similarly, given an allocation $X$ and a set of doctors $D'$, we define $X \cap D'$ as the allocation where each hospital $h$ receives the set $X_h \cap D'$.

We define the following ordering over allocations: given two non-redundant allocations $X$ and $Y$, we say that $X \succ Y$ if 
\begin{inparaenum}[(a)]
\item $\USW(X)> \USW(Y)$, or 
\item $\USW(X) = \USW(Y)$, and there is some doctor $d_j$ such that $X(d_j) \succ_{d_j} Y(d_j)$ and for all doctors $d_k$ where $k < j$ we have $X(d_k) = Y(d_k)$.
\end{inparaenum}
We call this ordering the {\em HWSD ordering}. Let $X$ be the output of Algorithm \ref{algo:red-jacket}; then for any other non-redundant allocation $Y$, we must have $X \succ Y$. Since this observation is used several times in our proofs, we state it formally. 

\begin{obs}\label{obs:hwsd-ordering}
Let $X$ be the output of Algorithm \ref{algo:red-jacket}, and let $X'$ be any other non-redundant allocation. Then $X \succ X'$ according to the HWSD ordering.
\end{obs}

To prove our result, we first show that it suffices to consider simple misreports, of the form $f_T$, when considering hospital deviations. Then we show that the outcome from reporting $f_T$ and the outcome from reporting $v_h$ cannot differ significantly in size. While the proof follows the broad structure of \citet{babaioff2021EF}, our proofs of Lemma \ref{lem:fTreport} and \ref{lem:manip-half} are more technically involved due to the two-sided nature of our problem.

\begin{lemma}\label{lem:fTreport}
Let $X$ be the output of Algorithm \ref{algo:red-jacket}.
If hospital $h$ reports $f_T$ instead of $v_h$ where $T \subseteq X_h$ and all the other hospitals report the same valuations, then Algorithm \ref{algo:red-jacket} outputs an allocation $Y$ where $Y_h = T$. 
\end{lemma}
\begin{proof}
When representing the valuation function $f_T$, we assume the graph $G_{f_T}$ used is a subgraph of $G_h$ (used for representing $v_h$). Indeed, it is always possible to represent $f_T$ using an appropriate subgraph of $G_h$ by keeping the edges of the bipartite graph representing $G_h$ that are in a maximum cardinality matching of doctors in $T$ to slots. For both allocations $X$ and $Y$, fix a matching of the allocated bundle for each hospital to the slots in their valuation graph; specifically for the hospital $h$, ensure the matching  under $f_T$ is a subset of the matching under $v_h$. 
Essentially, this means each doctor is not only associated with a hospital but also with a specific slot in the hospital's preference graph $G_h$. 
To make notation easier, we call the preference profile where every hospital $h$ reports $v_h$ as the {\em old} preferences, and the the profile where $h$ reports $f_T$ the \emph{new} preferences.

The proof is trivial for the case when $T = X_h$, since at each iteration, the algorithm makes the exact same decision under both profiles. 
It suffices to consider the case where $|T| = |X_h| - 1$, since we can then inductively assume $v_h=f_T$ and apply the same argument to $|T'|= |T|-1$ to obtain an arbitrary subset of $X_h$.

Let $d$ be the doctor in $X_h \setminus T$. 
Assume for contradiction that $|Y_h| < |T|$. 
We create an allocation $X'$ from $X$ by moving $d$ from hospital $h$ to the hospital it is allocated to in $Y$ --- hospital $j$ where $d$ is matched to some slot $q$. If slot $q$ is empty in the matching of allocation $X_j$ in the graph $G_j$, we add $d$ to slot $q$ of hospital $j$ and we stop. 
Otherwise, we swap $d$ with $d'$ in slot $q$ of $X_j$. We then repeat this process with $d'$. This transfer creates an alternating path of doctors and (hospital, slot) pairs; we call this path $P$ (see \Cref{fig:lem:fTreport} for an illustration). 
This path is finite because after every move, we strictly increase the number of (hospital, slot) pairs which have the same doctor allocated in both $X'$ and $Y$. 
We denote this path as $((h_1, q_1), d_1, (h_2, q_2), d_2, \dots, d_k, (h_{k+1}, d_{k+1}))$; the interpretation being that the doctor $d_i$ was moved from the hospital slot pair $(h_i, q_i)$ to $(h_{i+1}, q_{i+1})$ to create $X'$. 
\begin{figure}
    \centering
    \includegraphics[width=0.8\textwidth]{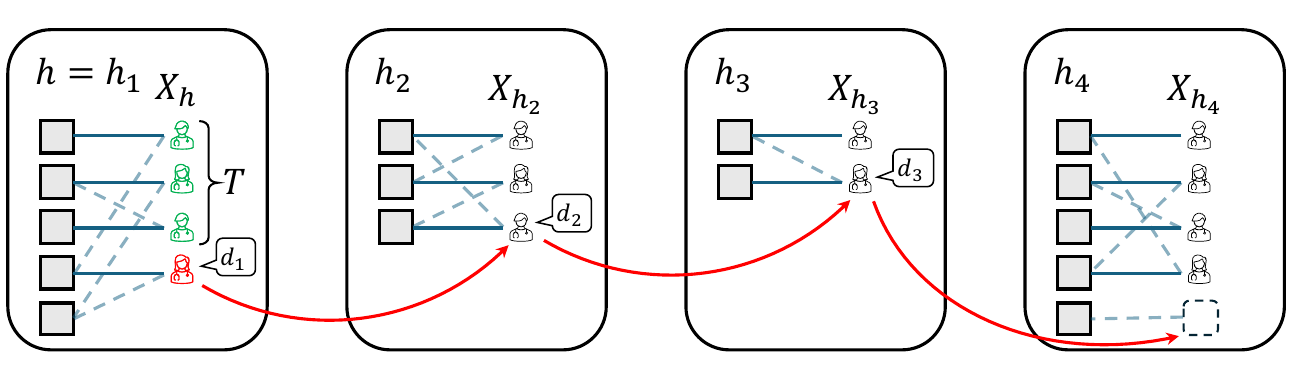}
    \caption{Illustration of the path transfer executed in the proof of \Cref{lem:fTreport}. Each hospital has availability slots (gray squares) that can be matched to different doctors (denoted by blue edges). Full blue edges denote the actual matching of doctors to slots, and dashed edges denote possible alternative slots. 
    The doctor $d_1$ is assigned to $h$ under $X_h$, but not under $Y$ (where $h$ receives $T$). We transfer $d_1$ to the slot+hospital it is assigned to under $Y$. That slot is occupied by the doctor $d_2$ who is moved to the spot of $d_3$. Finally, $d_3$ is assigned to their slot with $h_4$ under $Y$, which is unoccupied under $X$. This same path is reversed later in the proof to create the allocation $Y'$ from the allocation $Y$.}\label{fig:lem:fTreport}
\end{figure}
The resulting assignment $X'$ is non-redundant with respect to the new preferences. Moreover, $X' \ne Y$ since $X' = Y$ implies that $|Y_h| = |X'_h| = |T|$ which implies the lemma is satisfied. 
Therefore the allocation $X'$ must be strictly worse according to the HWSD ordering than $Y$ with respect to the new preferences (\Cref{obs:hwsd-ordering}). 
Let the doctors involved in the path $P$ be denoted by $D_P$. 
The set of doctors $D_P$ have the same allocation in both $X'$ and $Y$. 
Therefore, if $X'$ is worse than $Y$, it must be because of the doctors not present in the path. 
That is, the allocation $X' \setminus D_P$ is worse according to the HWSD ordering than $Y \setminus D_P$. 
Since this set of doctors has the same allocation in both $X$ and $X'$, we have that $X \setminus D_P$ is worse than $Y \setminus D_P$ according to the HWSD ordering.

Now, we reverse the path and apply it to $Y$ to get allocation $Y'$ which is non-redundant with respect to the old preferences. By reversing the path, we mean the following: if the path is $((h_1, q_1), d_1, (h_2, q_2), d_2, \dots, d_k, (h_{k+1}, d_{k+1}))$, we move $d_j$ from the hospital slot pair $(h_{j+1}, d_{j+1})$ to $(h_{j}, d_{j})$ to create $Y'$; note that this operation is well-defined. 
Reversing the path $P$ results in a non-redundant allocation (with respect to the old preferences) if no two doctors are matched to the same slot for some hospital. 
This happens because at each step of this reversed path we either 
\begin{enumerate}[(a)]
    \item allocate $d$ to some slot in $Y'_h$,
    \item swap a doctor allocated to some slot in a hospital with another doctor, or
    \item remove a doctor from a hospital.
\end{enumerate}
The second and third operations clearly preserve non-redundancy with respect to both the old and new preferences. 
The first operation preserves non-redundancy since we assume the matching of $Y_{h}$ in $G_h$ is a subset of the matching of $X_h$ in $G_h$. 
So, if the doctor $d$ was allocated to slot $q$ in $X_h$, then this slot must remain unoccupied in $Y_h$.

We observe that $Y' \ne X$ since $|Y'_h| \le |T| < |X_h|$. Therefore, $Y'$ must be strictly worse than $X$ according to the red jacket objective which implies that doctors not present in the path are worse off in $Y$ than in $X$ according to the red jacket objective; that is, the allocation $Y \setminus D_P$ is worse than the allocation $X \setminus D_P$ according to the HWSD ordering (\Cref{obs:hwsd-ordering}). This gives us a contradiction.
\end{proof}

Our next lemma shows that being honest does not significantly reduce the size of your bundle. 

\begin{lemma}\label{lem:manip-half}
Consider some hospital $h$. Fixing the reports of all other hospitals, let $X$ and $Y$ be the resulting allocation when $h$ reports $f_T$ (for some $T$) and $v_h$, respectively. If $v_h(T)=|T|$ and $X_h=T$, then $|Y_h| \ge \frac{|T|}2$.
\end{lemma}
\begin{proof}
We call the preference profile where $h$ reports $f_T$ and receives $X_h = T$ the \emph{misreported} preferences. 
For both allocations, fix a matching of the allocated bundle for each hospital to the slots in their valuation graph. 
Like the proof of \Cref{lem:fTreport}, when representing the valuation function $f_T$, we assume the graph $G_{f_T}$ used is a subgraph of $G_h$ (used for representing $v_h$). 

We also assume without loss of generality that the matching for the bundles $X_{h}$ and $Y_{h}$ is chosen in a manner that maximizes the number of slots that receive the same doctor under both matchings.  


Suppose that there are $k$ slots filled in $X_h$ which are not filled in $Y_h$, with doctors $\{d_1, \dots, d_k\}$ assigned to them. 
For each of these $k$ doctors we create a path as we did in \Cref{lem:fTreport}: we move the doctor $d_j$ to the hospital-slot pair it is assigned to in $Y$; if that slot is empty then we stop, otherwise we transfer the doctor in that slot to the slot they were assigned to under $Y$. Doing this, we create $k$ paths $P^{(1)}, \dots, P^{(k)}$. 
These paths are disjoint, since we fixed the matching of doctors to slots for each hospital in both $X$ and $Y$. 

Consider a path $P^{(j)}$ starting with a doctor $d_j$. If the path $P^{(j)}$ only involves slots of the hospital $h$, then transferring doctors along this path (as we did in Lemma \ref{lem:fTreport}) increases the number of slots in $G_h$ which receive the same allocation in both $X_h$ and $Y_h$, without changing the hospital $h$'s utility. This contradicts the assumption that the matchings we chose for $h$ under $X_h$ and $Y_h$ maximally intersect. 
Therefore, all paths must contain one hospital-slot pair that is not the hospital $h$.

Assume that transferring doctors along the path $P^{(j)}$ creates a non-redundant allocation $X'$ with respect to the misreported preferences (when hospital $h$ reported $f_T$).   
Using \Cref{obs:hwsd-ordering}, it must be the case that $X'$ is worse than $X$ according to the HWSD ordering. 
Since only doctors on the path $P^{(j)}$ had their assignment changed, this means $X' \cap D_{P^{(j)}}$ is worse than $X \cap D_{P^{(j)}}$ where $D_{P^{(j)}}$ is the set of doctors in the path $P^{(j)}$. This implies $Y \cap D_{P^{(j)}}$ is worse than $X \cap D_{P^{(j)}}$. If this is true, we can reverse the path and apply it to $Y$ to improve the allocation according to the HWSD ordering, which is a contradiction. This follows from arguments similar to the previous proof.

Therefore, for all of the paths, transferring along these paths creates an allocation $X'$ that is \textit{not non-redundant} with respect to the misreported preferences. The only way this can happen is if we move a doctor $d \in Y_h \setminus T$ to $X'_h$; this follows from the fact that the only hospital with a different set of preferences in both allocations is hospital $h$. Therefore, each of these paths has a doctor from $Y_h \setminus X_h$ and since the paths are disjoint, no two paths share a doctor.

Therefore, we have $|Y_h \setminus X_h| \ge k \ge |X_h| - |Y_h|$ or that $|Y_h| \ge |X_h|-|Y_h \setminus X_h| \ge |X_h| - |Y_h|$, which implies that $|Y_h| \ge \frac{|X_h|}{2}$. 
\end{proof}
Combining \Cref{lem:fTreport,lem:manip-half} proves our claim. 

\begin{restatable}{theorem}{thmredjacketsp}\label{thm:red-jacket-strategyproofness}
    \Cref{algo:red-jacket} is $2$-approximately hospital-strategyproof when hospitals have binary OXS valuations.
\end{restatable}
\begin{proof}
    Suppose that when hospital $h$ reports their valuation $v_h$ truthfully, the resulting output of \Cref{algo:red-jacket} is $X$, and hospital $h$ receives the bundle $X_h$. 
    Suppose that hospital $h$ reports the valuation $v'_h$ instead of $v_h$, and \Cref{algo:red-jacket} outputs the allocation $Y$. 
    Let $T \subseteq Y_h$ be a subset of doctors such that $|T| = v_h(T) = v_h(Y_h)$; intuitively, $T$ is the set of doctors in $Y_h$ that hospital $h$ actually wants to receive. 
    Suppose that $h$ reports $f_T$ instead of $v'_h$, and let the resulting allocation be $Z$. From \Cref{lem:fTreport} we know that $v_h(Z_h) = |T| = v_h(Y_h)$. 
    If $h$ reports $v_h$ and receives the bundle $X_h$, then from \Cref{lem:manip-half} we get that 
    \begin{align*}
        v_h(X_h) \ge \frac12 v_h(Z_h) \ge \frac12 v_h(Y_h).
    \end{align*}
    Therefore, \Cref{algo:red-jacket} is $2$-approximately hospital strategyproof.
\end{proof}


Note that this analysis is tight since Example \ref{ex:impossibility} describes an instance where all hospitals have binary OXS valuations; therefore, any mechanism that outputs hospital-\USW maximizing allocations can only be at most $2$-approximately hospital-SP. Restricting hospital preferences further to binary capped additive valuations, we show that the HWSD mechanism is exact strategyproof for hospitals. 

\begin{restatable}{lemma}{corredjacketadditive}\label{corr:red-jacket-SP-capped-additive}
When hospitals have binary capped additive valuations, the HWSD mechanism is hospital-SP.
\end{restatable}

\subsection{Serial Dictatorship}
Let us next explore the Serial Dictatorship mechanism (\Cref{algo:doctor-round-robin-revisited}). 
This is another simple mechanism which 
offers (weaker) approximate welfare and (stronger) strategyproofness guarantees.

We start with all hospitals having an empty bundle, i.e., $X_h = \emptyset$ for all $h \in H$. 
We order the doctors from $d_1$ to $d_m$. 
Each doctor $d_i$ goes down their ranking list, and {\em proposes} to hospitals in decreasing order of preferences. 
When $d_i$ proposes to the hospital $h$, if $\Delta_h(X_h,d_i) = 1$, then the hospital {\em accepts}. 
Else the hospital {\em rejects} and the doctor moves on to the next hospital. 

\begin{algorithm}[ht]
    \DontPrintSemicolon
    \caption{Doctor Serial Dictatorship}
    \label{algo:doctor-round-robin-revisited}
    \DontPrintSemicolon
    $X = (X_0, X_1, \dots, X_n) \gets (D, \emptyset, \dots, \emptyset)$\;
    
    \For{$i$ in $1$ to $m$}{
        \For{$j$ in $1$ to $n$}{
            Let $h$ be the $j$-th highest ranked hospital for doctor $d_i$\;
            \tcp{Doctor $d_i$ proposes to hospital $h$}
            \uIf{$\Delta_{h}(X_{h}, d_i) = 1$}{ 
                \tcp{Hospital $h$ accepts}
                $X_{h} \gets X_{h} + d_i$\;
                $X_0 \gets X_0 \setminus d_i$\;
                Break\;
            }
            \If{$\Delta_{h}(X_{h}, d_i) = 0$}{ 
                \tcp{Hospital $h$ rejects}
                Continue\;
            }
        }
    }
    \Return $X$\;

\end{algorithm}

\begin{theorem}\label{thm:serial-dictatorship}
The doctor serial dictatorship mechanism has the following properties:
\begin{enumerate}[(i)]
    \item The mechanism is doctor strategyproof
    \item The output allocation is stable
    \item The algorithm runs in polynomial time
    \item The output allocation is $2$-approximately max hospital-\USW.
    \item The mechanism is $2$-approximately hospital strategyproof.
\end{enumerate}
\end{theorem}

Again, we only show that \Cref{algo:doctor-round-robin-revisited} is 2-approximate hospital strategyproof, relegating all other proofs to Appendix~\ref{apdx:sd}. 

We fix a hospital $h$. In this proof, we will also assume for ease of analysis that the unallocated doctors go to hospital $h_0$ who has the valuation function $v_{0}(S) = |S|$.

Assume the allocation when all hospitals report truthfully is $X$. Assume that the hospital $h$ now misreports and let the resulting allocation be $Y$. 
In order to prove our claim, we only need to show that $v_h(Y_h) \le 2|X_h| = 2v_h(X_h)$. 

We distinguish between two types of changes to the allocation between $X$ and $Y$. Some changes are due to doctors not proposing to hospitals that they originally proposed to, and some changes are due to doctors being rejected by hospitals that previously accepted them. 

We say that a hospital $h'\in H$ was {\em denied} a doctor $d\in D$ if $d$ was allocated to $h'$ in $X$, but the doctor $d$ was allocated to a hospital in $Y$ that $d$ prefers to $h'$. 
The set of doctors that $h'$ was denied is denoted $D_{h'}$, and the set of doctors accepted by $h'$ in the generation of $Y$ that were denied from some other hospital is $D^a_{h'}$. Intuitively, $D_{h'}$ is the set of doctors that would have proposed to $h'$ if $h$ had truthfully reported its preferences but now did not propose to $h'$, and $D_{h'}^a$ is the set of doctors that $h'$ was able to `steal' from other hospitals due to the misreport by $h$.
Since doctors go down their preference lists when proposing to hospitals, every doctor $d \in D_{h'}^a$ prefers $h'$ to the hospital they were assigned to under $X$.

We also say that a doctor $d$ was {\em newly rejected} by hospital $h'$ if $d$ was allocated to $h'$ in $X$ but $d$ was rejected by $h'$ in the generation of $Y$. 
Let $R_{h'}$ be the set of doctors newly rejected by the hospital $h'$, and $R^a_{h'}$ be the set of doctors accepted by $h'$ under $Y$ that were newly rejected by some other hospital.

In what follows, we upper bound $|D_h^a|$ and $|R_h^a|$. This will almost immediately imply our desired upper bound. 
We first upper bound $|D_h^a|$.
\begin{lemma}\label{lemma:deny-accept-upper-bound}
$v_h(D^a_h \cup (X_h \cap Y_h)) \le |X_h|$.
\end{lemma}
\begin{proof}
By definition, all doctors in $D^a_h$ propose to $h$ in the generation of $X$, and were rejected by $h$; this is again because doctors propose to hospitals in decreasing order of preference. 
Recall that hospitals have matroid rank valuations, therefore, by the matroid augmentation property (\Cref{obs:augmentation}), if $v_h(S) < v_h(T)$ there is some doctor $d \in T \setminus S$ such that $\Delta_h(S,d) >0$. 
Therefore, if $v_h(D^a_h \cup (X_h \cap Y_h)) > |X_h| = v_h(X_h)$, there is some doctor $d$ in $D^a_h$ such that $\Delta_h(X_h, d) = 1$. This contradicts the definition of $D^a_h$ since the doctor $d$ would have been accepted by the hospital $h$ during the construction of the allocation $X$.
Therefore, it must be the case that $v_h(D^a_h \cup (Y_h \cap X_h)) \le |X_h|$.
\end{proof}

Next, we show that $|R_h^a| \le |X_h\setminus Y_h|$.
This proof is somewhat more involved, and requires the construction of an auxiliary set of doctors. 

Fix some $h' \in H \setminus \{h\}$.
We define $X^{t}_{h'}$ and $Y^{t}_{h'}$ as the bundles $X_{h'}$ and $Y_{h'}$ immediately after the iteration where doctor $d_t$ was allocated in \Cref{algo:doctor-round-robin-revisited}. 
We construct a new sequence of bundles $W_{h'}^t$ for $t \in [m] \cup \{0\}$. We initialize $W_{h'}^0=\emptyset$; then, we iterate through $t$ from $1$ to $m$. 
For each $t$, we set $W_{h'}^{t}\gets W_{h'}^{t-1} + d_t$ if:
\begin{enumerate}[(R1)]
    \item\label{rule:XandY} $d_t \in X_{h'} \cap Y_{h'}$. 
    \item\label{rule:XnotY} $d_t \in X_{h'} \setminus Y_{h'}$, and $Y^{t-1}_{h'} + d_t$ is redundant for ${h'}$, i.e., $\Delta_{h'}(Y^{t-1}_{h'},d_t) = 0$.
    \item\label{rule:D} $d_t \in D^a_{h'}$.
\end{enumerate}
The final bundle $W_{h'}^m$ contains $X_{h'}\cap Y_{h'}$ and $D_{h'}^a$, as well as any doctors in $X_{h'}\setminus Y_{h'}$ who would not have been selected by $h'$ during the run of \Cref{algo:doctor-round-robin-revisited} when $h$ misreports its preferences. 

Our basic proof strategy is to show that $W_{h'}^t$ is non-redundant and all doctors in $X_{h'} \setminus W_{h'}^m$ are in $D_{h'}$, i.e., they were denied from $h'$ due to the misreport of hospital $h$. 
We use this to obtain the required upper bound on $|R_{h'}|$. This proof will require the following useful matroid property.
\begin{restatable}{lemma}{lemmatroidproperty}\label{lem:useful-matroid-property}
Let $A$ and $B$ be two sets of doctors which are non-redundant with respect to some $v_h$. Let $d \notin A \cup B$ be some doctor such that $\Delta_h(A, d) = 1$ and $\Delta_h(B, d) = 0$. Then there is some doctor $d' \in B \setminus A$ such that $\Delta_h(A, d') = 1$; that is $A + d'$ is non-redundant with respect to $v_h$.
\end{restatable}
\begin{proof}
If $|B| > |A|$, then this lemma trivially holds via the matroid augmentation property. 

If $|B| \le |A|$, there must be some doctor $d^* \in (A + d) \setminus B$ such that $\Delta_h(B, d^*) = 1$, by the matroid augmentation property applied to the sets $B$ and $A+d$. Note that $d^*\ne d$ by our assumption in the lemma.
Set $B' = B + d^*$ and repeat this process with $B'$ until $|B'| = |A| + 1$. 

There must be some element $d' \in B' \setminus A$ such that $\Delta_h(A, d') = 1$, again by invoking the matroid augmentation property with the sets $B'$ and $A$. 
This element $d'$ must also be in $B \setminus A$ by the construction of $B'$.
\end{proof}

We start by showing that $W_{h'}^t$ is non-redundant with respect to the valuation of hospital $h'$.

\begin{lemma}\label{lem:w-non-redundant}
At any $h' \in H$ and any $t \in [m] \cup \{0\}$, $W_{h'}^{t}$ is non-redundant with respect to $v_{h'}$.
\end{lemma}
\begin{proof}
We prove this by induction. At $t=0$, $W_{h'}^0 = \emptyset$ and the lemma trivially holds. 
Assume that $W_{h'}^{t'}$ is non-redundant for all $t' < t$. We show that it holds for $t$.

We first show that if Rule \ref{rule:D} is applied, i.e., $d_t \in D_{h'}^a$, then $W_{h'}^{t}$ is non-redundant. 
Assume for contradiction that $W_{h'}^{t-1} + d_t$ is redundant. 
Since $d_t \in D_{h'}^a \subseteq Y_{h'}$ and $Y_{h'}^{t-1} \subseteq Y_{h'}$ we have that $Y^{t-1}_{h'} + d_t$ is non-redundant. Therefore, there must be some doctor $d_j \in W_{h'}^{t-1} \setminus Y^{t-1}_{h'}$ such that $\Delta_{h'}(Y^{t-1}_{h'}, d_j) = 1$ (by \Cref{lem:useful-matroid-property}). 

Since $d_j$ is not in $Y_{h'}^{t-1}$, and $Y_{h'}^j \subseteq Y_{h'}^{t-1}$, we must have that $d_j$ was not added to $Y_{h'}$ at the $j$-th iteration and $d_j \notin Y_{h'}$. 
In particular, since Rules \ref{rule:XandY} and \ref{rule:D} only add doctors that are in $Y_{h'}$, they do not apply to $d_j$. The only way that $d_j$ could have been added to $W_{h'}^{j-1}$ is if Rule \ref{rule:XnotY} applies.
Rule \ref{rule:XnotY} requires that $d_j$ is redundant for $Y_{h'}^{j-1}$. 
Since $v_{h'}$ is submodular and $Y_{h'}^{j-1} \subseteq Y_{h'}^{t-1}$, $\Delta_{h'}(Y_{h'}^{t-1},d_j) = 0$ as well. Thus, $d_j$ cannot have been added to $W_{h'}^{j-1}$ and is not in $W_{h'}^{t-1}$, a contradiction. Thus, $W_{h'}^{t-1}+d_t$ is non-redundant if $d_t$ is added via Rule \ref{rule:D}. 

Next, suppose that we add some doctor $d_t$ such that $d_t \in X_{h'}$. The only way that this can occur is if $d_t$ is added by applying Rules \ref{rule:XandY} and \ref{rule:XnotY}.
Assume for contradiction that $W_{h'}^{t-1} + d_t$ is redundant. 
By \Cref{lem:useful-matroid-property}, there must be some doctor $d_j \in W_{h'}^{t-1} \setminus X_{h'}$ such that $j < t$ and $X_{h'} - d_t + d_j$ is non-redundant. 
The doctor $d_j$ is not in $X_{h'}$ and thus it was not added due to Rules \ref{rule:XandY} or \ref{rule:XnotY}. 
Thus, $d_j$ was added due to Rule \ref{rule:D} and is in $D^a_{h'}$. Therefore, $d_j$ must have proposed to $h'$ during the generation of allocation $X$, and was rejected. 
Recall that $X_{h'}^{j-1} \subseteq X_{h'}$; thus, if $X_{h'} - d_t + d_j$ is non-redundant, then $X^{j-1}_{h'} + d_j$ is non-redundant by submodularity of $v_{h'}$, and $h'$ should have accepted $d_j$ in the generation of $X$. This is a contradiction and therefore, $W_{h'}^t$ must be non-redundant. 
\end{proof}
\Cref{lem:w-non-redundant} immediately implies an upper bound on the size of $W_{h'}^m$.
\begin{lemma}\label{lem:W-lessthan-Y}
    $|W_{h'}^m|\le |Y_{h'}|$
\end{lemma}
\begin{proof}
According to \Cref{lem:w-non-redundant}, $W_{h'}^m$ is non-redundant with respect to $v_{h'}$. Therefore, if $|W_{h'}^m| > |Y_{h'}|$, there must be at least one doctor $d \in W_{h'}^m \setminus Y_{h'}$ such that $\Delta_{h'}(Y_{h'}, d) = 1$. 
This doctor $d$ has to satisfy $d \notin Y_{h'}$.
The only way that a doctor not in $Y_{h'}$ is added to $W_{h'}^m$ is via the application of Rule \ref{rule:XnotY}. 
However, Rule \ref{rule:XnotY} requires that $\Delta_{h'}(Y_{h'},d) = 0$; thus, $d$ should have never been added to $W_{h'}^m$ in the first place.
\end{proof}

\begin{lemma}\label{lem:rh-upperbound}
For all $h' \ne h$, $|R_{h'}| \le |R^a_{h'}|$.
\end{lemma}
\begin{proof}
From the construction of $W_{h'}^m$, every doctor in $X_{h'}$ that is not in $W_{h'}^m$ must be in $D_{h'}$ according to Rule \ref{rule:XnotY}. 
Therefore, $X_{h'} \setminus W_{h'}^m \subseteq D_{h'}$; in particular, $|X_{h'}| - |X_{h'} \cap W_{h'}^m| \le |D_{h'}|$. 
Recall that every doctor in $D_{h'}^a$ is also in $W_{h'}^m$ according to Rule \ref{rule:D}. In other words, $W_{h'}^m$ contains all doctors that were denied from some other hospital when $h$ misreports, and were assigned to $h'$ as a result. 
Thus, $W_{h'}^m \setminus X_{h'} = D^a_{h'}$; in particular, $|W_{h'}^m| - |X_{h'} \cap W_{h'}^m| = |D^a_{h'}|$. 

Finally, by \Cref{lem:W-lessthan-Y}, $|W_{h'}^m| \le |Y_{h'}|$.
Combining these observations we get:
\begin{align}
    |X_{h'}| - |Y_{h'}| &\le |X_{h'}| - |W_{h'}^m| = |X_{h'}| - |X_{h'}\cap W_{h'}^m| + |X_{h'}\cap W_{h'}^m|- |W_{h'}^m| \notag\\
    &= |X_{h'}| - |X_{h'}\cap W_{h'}^m| - (|W_{h'}^m|-|X_{h'}\cap W_{h'}^m|)\le |D_{h'}| - |D^a_{h'}|.\label{eq:X-YleD-Da}
\end{align}
Thus, $|X_{h'}| - |Y_{h'}| \le |D_{h'}| - |D^a_{h'}|$.
Consider the set of doctors $X_{h'}$ assigned to $h'$ when the hospital $h$ truthfully reports its valuation.
It comprises of three disjoint sets of doctors: 
\begin{inparaenum}[(a)]
\item doctors that are assigned to $h'$ under both $X_{h'}$ and $Y_{h'}$, i.e., $X_{h'} \cap Y_{h'}$
\item doctors that were denied from $h'$, i.e. the set $D_{h'}$, and
\item doctors that were newly rejected by $h'$, i.e. the set $R_{h'}$.  
\end{inparaenum}
Thus, $X_{h'} = (X_{h'} \cap Y_{h'}) \cup D_{h'} \cup R_{h'}$. Similarly, we can show that $Y_{h'} = (X_{h'} \cap Y_{h'}) \cup D^a_{h'} \cup R^a_{h'}$. 
Therefore, 
\begin{align}
|X_{h'}| - |Y_{h'}| &= |X_{h'} \cap Y_{h'}| +|D_{h'}| +|R_{h'}|-\left(|X_{h'} \cap Y_{h'}| +|D^a_{h'}| +|R^a_{h'}|\right)\notag\\
&= |D_{h'}| - |D_{h'}^a| + |R_{h'}| - |R_{h'}^a|.\label{eq:X-Y-DDaRRa-bound}
\end{align}
Plugging in the upper bound in \Cref{eq:X-YleD-Da} into \Cref{eq:X-Y-DDaRRa-bound} we get  
\begin{align*}
    |D_{h'}| - |D_{h'}^a| + |R_{h'}| - |R_{h'}^a| \le |D_{h'}| - |D_{h'}^a| \Rightarrow |R_{h'}| \le |R_{h'}^a|
\end{align*}
which concludes the proof.
\end{proof}

We are now ready to prove the final lemma of the theorem.
\begin{lemma}\label{lem:rah-upperbound}
$|R^a_h| \le |X_h \setminus Y_h|$.
\end{lemma}
\begin{proof}
This proof follows from Lemma \ref{lem:rh-upperbound}. 
For any $h'\ne h$, we note that each doctor in $R^a_{h'}$ is either from $R_{h''}$ for some $h''$ or $X_h \setminus Y_h$. 
Therefore, we have $\bigcup_{h' \ne h} R^a_{h'} \subseteq (\bigcup_{h' \ne h} R_{h'} ) \cup (X_h \setminus Y_h)$. 
By re-arranging terms, we conclude that 
\begin{align}
    \left(\sum_{h' \ne h} |R^a_{h'}|\right) - \left(\left|\left(\bigcup_{h' \ne h}R^a_{h'}\right) \cap \left(\bigcup_{h' \ne h} R_{h'}\right)\right|\right) \le |X_h \setminus Y_h| \label{eq:rah-upper}
\end{align}

To prove the lemma, we first use the fact that $R^a_{h} \subseteq \bigcup_{h' \ne h} R_{h'} \setminus \bigcup_{h' \ne h} R^a_{h'}$. This follows from the definition of $R^a_h$ as doctors which have been newly rejected by other hospitals but accepted by $h$. We can therefore write,

\begin{align*}
    |R^a_{h}| &\le \sum_{h' \ne h}|R_{h'}| - \left|\left(\bigcup_{h' \ne h}R^a_{h'}\right) \cap \left(\bigcup_{h' \ne h} R_{h'}\right)\right| \le \sum_{h' \ne h}|R^a_{h'}| - \left|\left(\bigcup_{h' \ne h}R^a_{h'}\right) \cap \left(\bigcup_{h' \ne h} R_{h'}\right)\right|  \le |X_h \setminus Y_h|
\end{align*}
The second inequality follows from the fact that $|R_{h'}| \le |R^a_{h'}|$ for all $h' \ne h$, as per \Cref{lem:rh-upperbound}. The third inequality follows from \eqref{eq:rah-upper}.
\end{proof}

\begin{restatable}{theorem}{thmsdsp}\label{thm:roundrobin-hospital-sp}
\Cref{algo:doctor-round-robin-revisited} is $2$-approximately strategyproof for the hospitals.
\end{restatable}
\begin{proof}
Assume that hospital $h\in H$ misreports its valuation, and that the resulting output of \Cref{algo:doctor-round-robin-revisited} is now $Y$ instead of $X$. 
According to \Cref{lemma:deny-accept-upper-bound}, $v_h(D_{h}^a \cup (X_h \cap Y_h))\le |X_{h}|$. According to \Cref{lem:rah-upperbound}, $|R_{h}^a| \le |X_{h}\setminus Y_{h}|$. 

If a hospital $h'$ has a doctor $d$ in $Y_{h'}\setminus X_{h'}$, then $d$ must either be denied from another hospital and assigned to $h'$, i.e., it is in $R^a_{h'}$; alternatively, $d$ is in $D^a_{h'}$, in which case $d$ prefers $h'$ to whichever hospital they were assigned to under $X$. 
We can partition $h$'s assignment when it misreports into $Y_h = (Y_h \cap X_h) \cup D^a_{h} \cup R^a_{h}$. Using this, we get
\begin{align*}
    v_h(Y_h) &= v_h((X_h \cap Y_h) \cup D_h^a \cup R_h^a) \le v_h((X_h \cap Y_h) \cup D_h^a) + |R_h^a| \le |X_h| + |X_h \setminus Y_h| \le 2|X_h|.
\end{align*}
which concludes the proof.
\end{proof}

This analysis is tight, once again due to Example \ref{ex:impossibility}. The same misreport described in the example shows that we cannot beat $2$-approximate hospital-SP with the serial dictatorship mechanism.

\section{Doctors with Cardinal Utilities}\label{sec:cardinal-utilities}
We now turn our attention to a more general problem where hospitals still have matroid rank valuations but doctors have cardinal valuations over the doctors. These valuations are denoted by the function $c_d: H \rightarrow \R$; $c_d(h)$ denotes the utility that doctor $d$ derives from being assigned to the hospital $h$. Note that cardinal preferences can encode both ties and incomplete orders in doctors preferences; we assume setting $c_d(h) < 0$ implies that the doctor $d$ prefers to be unallocated than allocated to hospital $h$. 

In this section, we continue to define stability using \Cref{def:stability}. Specifically, we note that \Cref{def:blocking-pair} continues to be well-defined and \Cref{obs:blocking} holds in the presence of ties.
\subsection{Maximizing Doctor Welfare}
We first show that the problem of maximizing doctor-\USW is an instance of the weighted matroid intersection problem which admits polynomial time algorithms. The maximum doctor-\USW allocation among all non-redundant allocations is guaranteed to be stable, so this offers us an efficient algorithm to compute a stable, doctor-\USW maximizing allocation. 
We can in fact make a slightly stronger statement.
\begin{restatable}{lemma}{lemmaxdoctorwelfare}\label{lem:max-doctor-welfare-max-hospital-USW}
    There exists a polynomial time algorithm that computes a stable, max doctor-\USW allocation. Moreover, the output allocation has the highest possible hospital-\USW among all stable, max doctor-\USW allocations.
\end{restatable}
\begin{proof}
The ground set of both matroids is the set of all doctor hospital pairs $E = \{(d, h): d \in D, h \in H\}$. In the first matroid, a set $S$ is independent if the allocation $X$ where $X_h = \{d: (d, h) \in S\}$ is non-redundant.
The second matroid is a partition matroid where each part is of the form $\{(d, h):h \in H\}$ for each $d \in D$. This matroid ensures that each doctor is allocated at most once. Each element $(d, h)$ has weight $c_d(h)$.

It is easy to see that the max weighted independent set at the intersection of these two matroids corresponds to the non-redundant allocation with maximum doctor-\USW. For this problem there exist efficient algorithms to compute the max weighted independent set subject to the independent set having size equal to some integer $k$ \cite{brezovec1986matroidintersection}. 
The size of the independent set is equal to the hospital-\USW of the output allocation. Therefore, by trying all $m$ possible values of $k$, we can find the max doctor-\USW allocation with the highest possible hospital-\USW.
\end{proof}

By setting $c_d(h)$ values appropriately, the above result can be used to compute both the outputs of Algorithms \ref{algo:red-jacket} and \ref{algo:doctor-round-robin-revisited}.
If doctors have ordinal preferences, we use $u_d(h)$ to denote the Borda score of hospital $h$ according to doctor $d$, i.e., if $d$ ranks $h$ at the $k$-th position, then $u_d(h) = n-k$. 
If we instantiate Lemma \ref{lem:max-doctor-welfare-max-hospital-USW} with $c_d(h) = (2n + 1)^{n-d}{u_d(h)}$, we recover the output of the doctor Serial Dictatorship mechanism (\Cref{algo:doctor-round-robin-revisited}). 
If we set $c_d(h) = M + (2n + 1)^{n-d}{u_d(h)}$, where $M$ is a very large number, we recover \Cref{algo:red-jacket}. 

Assuming we receive cardinal valuations for each doctor $c_d$, we can compute a doctor-\USW maximizing non-redundant allocation. Moreover, since each doctor gets allocated to at most one hospital, by instantiating Lemma \ref{lem:max-doctor-welfare-max-hospital-USW} with valuations $c'_d(h) = \log c_d(h)$, we compute a non-redundant doctor-\NSW maximizing allocation. Both these allocations are stable because, if either allocation is not stable, there is a hospital $h$ and a doctor $d$ such that we can move the doctor $d$ to hospital $h$ to increase its utility while not affecting any other doctor's utilities (\Cref{obs:blocking}). This contradicts the fact that we picked the doctor-\USW (or doctor-\NSW) maximizing allocation.

\subsection{Doctor-Optimal and Hospital Nash Optimal Allocations}
In Lemma \ref{lem:max-doctor-welfare-max-hospital-USW}, we showed that subject to maximizing doctor-\USW, we can maximize hospital-\USW as well. In this section we show that subject to maximizing doctor-\USW, we can maximize hospital-\NSW. By appropriately setting doctor valuations, we can prove a similar result with doctor-\NSW maximizing allocations. 

We say that an allocation $X$ is $k$-maximal doctor-\USW ($k$-MDW) if it maximizes doctor-\USW, subject to hospitals having a \USW of exactly $k$. Such an allocation can be computed in polynomial time (using a variant of Lemma \ref{lem:max-doctor-welfare-max-hospital-USW}).

Our algorithm is based on \emph{utility capping} and {\em local search}. 
Given a hospital $h \in H$ with a valuation function $v_h$ and an integer $c >0$, we define the \emph{capped valuation function} of $h$ as $v^c_h(S) = \min\{v_h(S), c\}$. We note that if $v_h$ is a matroid rank function, so is $v_h^c$. 
The valuation $v^c_h$ limits hospital $h$ to a maximum of $c$ doctors. 

Our algorithm works as follows: for each value of $k$, we find a $k$-MDW allocation $X$. 
Next, we cap the utility of every hospital $h$ at $c_h = |X_h|$. 

We utilize a local search technique to find a doctor-\USW optimal/hospital-\NSW optimal allocation, subject to maintaining total hospital welfare at $k$. 
We identify two hospitals $i$ and $j$ such that $|X_i| \ge |X_j| +2$ or $|X_i| = |X_j| + 1$ and $j < i$, and there is a $k$-MDW allocation $X'$ such that 
\begin{enumerate}[(i)]
    \item $|X'_i| = |X_i| - 1$,
    \item $|X'_j| = |X_j| + 1$, and 
    \item $|X'_p| = |X_p|$ for all $p \ne i, j$.
\end{enumerate}
We can efficiently check whether such an allocation $X'$ exists by increasing the cap of hospital $j$, $c_j$, by 1 and decreasing the cap of hospital $i$, $c_i$, by 1 before computing a $k$-MDW allocation. 
We check if this cap change reduces the doctors' welfare; if so, we reject the new allocation and try a different pair of candidate hospitals $i, j\in H$. 

If there exists an allocation $X'$ where hospitals $i$ and $j$'s caps can be adjusted, we replace $X$ with $X'$ and repeat the process. 
We stop when an update is no longer possible. 
We do this for each value of $k \in [m]$ to generate $m$ allocations $X^1,\dots,X^m$; we then simply pick one that maximizes doctor-\USW, and subject to that maximizes the hospital-\NSW.  
The steps are described in \Cref{algo:max-nash-welfare}.

\begin{algorithm}[t]
    \DontPrintSemicolon
    \caption{Maximizing Hospital Nash Welfare subject to Maximizing Doctor Welfare}
    \label{algo:max-nash-welfare}
    \DontPrintSemicolon
    \For{$k$ in $1$ to $m$}{
        $X = (X_0, X_1, \dots, X_n) \gets$ a $k$-MDW allocation\;
        $c_i \gets |X_i|$ for each $i \in H$\;
        \While{1}{
        \For{$i$ in $1$ to $n$}{
            \For{$j$ in $1$ to $n$}{
                \If{$|X_i| \ge |X_j| + 2$ or $|X_i| = |X_j| + 1$ and $j < i$}{
                    $c'_p = 
                    \begin{cases}
                        c_i - 1 & p = i \\
                        c_j + 1 & p = j \\
                        c_p & p\ne i, j
                    \end{cases}$\;
                    $X' \gets$ a $k$-MDW allocation with respect to the capped valuations\;
                    \If{$X'$ is a $k$-MDW allocation with respect to the original valuations}{
                        $X \gets X'$\;
                        $\vec c \gets \vec c'$\;
                    }
                }
            }
        }
        \If{$X=X'$}{
                    Set $X^k \gets X$\;
                    Break out of the while loop\;
                }
        }
    }
Return the allocation $X^k$ with the highest doctor-\USW and subject to that, the one that maximizes $\prod_{h \in H} \left (v_h \left (X^k_h\right ) \right )$\;
\end{algorithm}

The proof of \Cref{algo:max-nash-welfare}'s correctness (\Cref{thm:max-nash-welfare}) relies on the following key lemma.
\begin{lemma}\label{lem:max-nash-welfare-main}
Let $X$ and $Y$ be two $k$-MDW allocations. Then, for any hospital $i \in H$ such that $|X_i| < |Y_i|$, there exists a hospital $j \in H$ such that $|X_j| > |Y_j|$ and a $k$-MDW allocation $X'$ such that 
\begin{enumerate}[(i)]
    \item $|X'_i| = |X_i| + 1$,
    \item $|X'_j| = |X_j| - 1$, and 
    \item $|X'_p| = |X_p|$ for all $p \ne i, j$.
\end{enumerate}
\end{lemma}

This lemma almost immediately gives us the main proof of this section (proof in Appendix \ref{apdx:cardinal-utilities}). 
\begin{restatable}{theorem}{thmmaxnashwelfare}\label{thm:max-nash-welfare}
\Cref{algo:max-nash-welfare} terminates in polynomial time and outputs an allocation $X$ that maximizes doctor welfare and subject to that maximizes hospital Nash welfare.
\end{restatable}

When the allocation that maximizes doctor-\USW is unique, then the above theorem is not meaningful. However, when there are many ties in the optimal doctor-\USW allocation, the above theorem suggests that we can find a fair one.
\subsection{Proof of Lemma \ref{lem:max-nash-welfare-main}}

The problem of maximizing doctor-\USW can be modeled as a weighted matroid intersection problem over the ground set $E = \{(h, d) | d \in D, h \in H\}$. Any assignment of doctors to hospitals can be equivalently defined as a subset of the ground set $E$. Specifically, each allocation $X$ can be denoted by the set $\bigcup_{d \in D} (X(d), d)$, which is a subset of $E$. Similarly, any subset of $E$ can be written as an allocation, although this allocation might allocate a doctor multiple times.

Given an allocation $X$ represented as a subset of the ground set $E$, we construct the \emph{matroid exchange graph} $\cal G(X)$. $\cal G(X)$ is a directed bipartite graph over the set of nodes $E$. 
The nodes of this graph are hospital-doctor pairs, and there are two types of edges. 
An edge exists from the node $(h, d) \in X$ to $(h', d') \in E \setminus X$ if swapping $(h, d)$ with $(h', d')$ in the allocation $X$ does not violate the non-redundancy of the allocation. In other words, withholding $d$ from $h$ and assigning $d'$ to $h'$ instead results in a non-redundant allocation. 
An edge exists from $(h', d') \in E \setminus X$ to $(h, d)\in X$ if swapping the two in the allocation $X$ does not result in any doctor being allocated to multiple hospitals. 
Any edge from $(h, d) \in X$ to $(h', d') \in E \setminus X$ has a weight of $c_d(h) - c_{d'}(h')$. 
All other edges have a weight of $0$.

For any $S \subseteq X$ and $S' \subseteq E \setminus X$ such that $|S'| = |S|$, a matching $M$ is a set of $|S|$ edges in $\cal G(X)$ from $S$ to $S'$ such that every node in $S \cup S'$ is connected to at least one edge in $M$. A back matching $M'$ is a set of $|S|$ edges from $S'$ to $S$ such that every node in $S \cup S'$ is connected to at least one edge. We say the set $(S, S')$ is a valid swap if $(X \setminus S) \cup S'$ is a valid allocation --- it is non-redundant and each doctor is allocated at most once. We refer to the allocation $(X \setminus S) \cup S'$ as the allocation generating by {\em performing} the valid swap $(S, S')$ on the allocation $X$.
We slightly abuse notation if the matching and back matching form a cycle $C$: we say the cycle $C$ is a valid swap if the nodes in the cycle form a valid swap. We similarly define the act of performing the valid swap $C$ on an allocation $X$.

We utilize the following results on exchange graphs from \citet{brezovec1986matroidintersection}. 

\begin{theorem}[\citet{brezovec1986matroidintersection}]\label{thm:brezovec}
Let $X$ and $Y$ be two $k$-MDW allocations
\begin{enumerate}[(a)]
    \item  There exists a matching and a back matching in the exchange graph $\cal G(X)$ between $X\setminus Y$ and $Y \setminus X$.
    \item There is no negative weight cycle in the exchange graph $\cal G(X)$. 
    \item Consider any sets $S \subseteq X$ and $S' \subseteq E \setminus X$ such that $|S| = |S'|$. Assume there is a matching and back-matching between these two sets in $\cal G(X)$. Then either $(S, S')$ is a valid swap, or there is another matching and back matching between these two sets which are not identical to the first matching and back matching respectively. 
    \item Let there be two matchings and back matchings between two sets $S \subseteq X$ and $S' \subseteq E \setminus X$ such that each matching and back matching corresponds to a single cycle. We can decompose these two cycles (say $C$ and $C'$) into a set of cycles $C_1, C_2, \dots, C_u$ such that each new cycle has length smaller than $2|S|$. Moreover, each edge appears the same number of times in $C_1, \dots, C_u$ as it appears in the two cycles $C$ and $C'$.
\end{enumerate}
\end{theorem}
All four statements above can be found in \citet{brezovec1986matroidintersection}: (a) is equivalent to Lemma 2, (b) is equivalent to Theorem 3, (c) is equivalent to Lemma 3, and (d) is equivalent to Lemma 4.

Consider any two $k$-MDW allocations $X$ and $Y$. 
From Theorem \ref{thm:brezovec}(a), a matching and back-matching must exist between $X \setminus Y$ and $Y \setminus X$. 
The matching and back matching form a set of disjoint cycles $\cal C$. 

We say an edge from $(h, d) \in X$ to $(h', d') \in E \setminus X$ is a {\em cross edge} if $h \ne h'$. 
This essentially means that $\Delta_{h'}(X_{h'}, d') = 1$. 
Therefore, there must exist an edge from all $(h'', d'') \in X$ to $(h', d')$ in $\cal G(X)$. We use this observation for our next operation, which we call {\em uncrossing}. 

\begin{lemma}[Uncrossing Lemma]\label{lem:uncrossing}
For some allocation, let $C$ be a cycle in the exchange graph $\cal G(X)$ over the set of nodes $S$. There exist a set of disjoint cycles $C_1, C_2, \dots, C_u$ in $\cal G(X)$ such that each node in $S$ is contained in exactly one cycle and each cycle has at most one cross edge. 
\end{lemma}
\begin{proof}
Assume that the cycle $C$ has more than one cross-edge. 
Suppose that $C$ consists of the ordered set of edges $\{e_1, e_2, \dots, e_s\}$, where $e_1$ is a cross edge and the next cross edge after $e_1$ is at $e_t$. Let $e_1$ be an edge from $(h_1, d_1)$ to $(h_2, d_2)$, and let $e_t$ be an edge from $(h_3, d_3)$ to $(h_4, d_4)$. 

We replace $e_1$ with $e'_1$ which is an edge from $(h_1, d_1)$ to $(h_4, d_4)$. We replace $e_t$ with $e'_t$ which is an edge from $(h_3, d_3)$ to $(h_2, d_2)$. Both $e'_1$ and $e'_t$ must exist in $\cal G(X)$ since $e_1$ and $e_t$ are cross edges. We then output the cycles $C_1 = \{e_2, e_3, \dots, e_{t-1}, e'_t\}$ and $C_2 = \{e'_1, e_{t+1}, \dots, e_s\}$. 
$C_1$ has at most one cross edge; the only possible cross edge is $e_{t}'$. We then repeat this procedure with $C_2$ till we have a disjoint set of cycles, each with at most one cross-edge.
\end{proof}

\begin{proof}[Proof of Lemma \ref{lem:max-nash-welfare-main}]
From Theorem~\ref{thm:brezovec}(a), there is a matching and back matching between $X \setminus Y$ and $Y \setminus X$ in the exchange graph $\cal G(X)$. 
This matching and back matching can be assumed to be a set of cycles $\cal C$ with at most one cross edge in each cycle (Lemma \ref{lem:uncrossing}). 

We associate each cycle with a winning and losing hospital. Suppose that the cycle $C$ corresponds to a valid swap. A hospital is winning with respect to a cycle $C$ if its utility strictly increases if the swap is executed. A losing hospital is one whose utility strictly decreases if this happens. 
Since each cycle has at most one cross edge, there is at most one winner and loser in a cycle. Even if $C$ is not a valid swap, we define winners and losers by pretending it is a valid swap and looking at the size of each hospital's bundle before and after the swap. 

The total weight of all the cycles in $\cal C$ is $0$. This follows from the fact that both $X$ and $Y$ are $k$-MDW allocations, and the total weight of all the cycles in $C$ is equal to $\sum_{(h, d) \in X \setminus Y} c_d{(h)} - \sum_{(h', d') \in Y \setminus X} c_{d'}{(h')}$.
Since there is no negative weight cycle in the graph (Theorem \ref{thm:brezovec}(b)), each cycle in $\cal C$ has a weight of $0$. 
By the assumptions made in the lemma, $|Y_i| > |X_i|$; so there must be (at least) one cycle $C$ in $\cal C$ where hospital $i$ is a winner. We show that either $C$ is a valid swap, or we can find a smaller $0$ weight cycle where $i$ is the winner. 

If $C$ is a valid swap, then we are done. If $C$ is not a valid swap, then from Theorem \ref{thm:brezovec}(c), there is another matching and back matching between the nodes of $C$. We then apply the uncrossing operation to this new matching and back matching to create a new set of cycles $C_1, \dots, C_s$.  Since $C$ is a 0 weight cycle, each of $C_1, \dots, C_s$ must also be a $0$ weight cycle (Theorem \ref{thm:brezovec}(b)).
\begin{description}[leftmargin=0cm]
    \item[Case 1: $s > 1$.] There must be a cycle $C'$ in this new set of cycles such that $i$ is the winner. Note that $C'$ is smaller than $C$ and $C'$ has $0$ weight.
\item[Case 2: $s = 1$.] We invoke Theorem \ref{thm:brezovec}(d) to rewrite these two cycles $C$ and $C_1$ as a set of cycles $C'_1, \dots, C'_u$ where each $C'_w$ is smaller than $C$. We then apply the uncrossing lemma to each of these cycles to ensure that there is at most one cross edge in each of the cycles. Then we find the cycle $C'$ in this new set of cycles where $i$ is the winner. $C'$ is smaller than $C$. $C'$ also has $0$ weight since both $C$ and $C_1$ have $0$ weight, and there are no negative weight cycles in $\cal G(X)$ (Theorem \ref{thm:brezovec}(b)). 
\end{description}

In summary, either the cycle we find is a valid swap and has $0$ weight, Or we find a smaller cycle with $0$ weight where $i$ is the winner. We can repeat the above procedure with the smaller cycle until we find a cycle $C^*$ which is a valid swap. We perform this valid swap using $C^*$ on the allocation $X$ to create a $k$-MDW allocation $Z$ where $i$'s utility increases by $1$, some hospital $q$'s utility decreases by $1$ and all the other hospitals have the same utility. 

If $|X_q| > |Y_q|$, we are done. Otherwise, we repeat this procedure with the allocation $Z$ and $q$. This procedure is guaranteed to terminate, since with each new allocation we strictly increase $|X \cap Y|$; that is $|Z \cap Y| > |X \cap Y|$.
\end{proof}

\subsection{Hardness of Maximizing Hospital-\NSW subject to Stability}
We show that finding an approximately max hospital-\NSW allocation subject to stability is computationally intractable. 
The statement we show is in fact stronger than that: 

\begin{restatable}{theorem}{thmnashhardness} \label{thm:nash_hardness}
The problem of deciding whether there exists a stable allocation where all hospitals are allocated at least one doctor is NP-complete when hospitals have capped binary valuations.
\end{restatable}
\begin{proof}
We reduce from the NP-complete 2P2N-3SAT problem \citep{berman2003approx-hardness-max3sat}. This class comprises of 3SAT instances with $m$ clauses over variables $(x_1, \dots, x_n)$, such that each variable appears \emph{exactly} two times in positive form and two times in negative form. An instance is a ``yes'' instance if and only if it has a valid truth assignment. 
\Cref{eq:2p2nex} is an example instance of 2P2N-3SAT with $n = 3$ variables and $m = 4$ clauses. 
\begin{align}
    \overbrace{(x_1\vee x_2 \vee \neg x_3)}^{C_1}\wedge \overbrace{(\neg x_1 \vee x_2 \vee x_3)}^{C_2} \wedge \overbrace{(\neg x_1 \vee \neg x_2 \neg x_3)}^{C_3}\wedge \overbrace{(x_1 \vee \neg x_2 \vee x_3)}^{C_4}\label{eq:2p2nex}
\end{align}

Given a 2P2N-3SAT instance, our reduction is as follows: we have $m + 3n$ hospitals and $12n$ doctors. 
We create for each variable $x_i$, two copies of a doctor corresponding to the positive literal of the variable ($x_i, x_i'$), and two copies of the doctor corresponding to the negative literal of the variable ($\neg x_i, \neg x_i'$). 
We also create eight dummy doctors for each variable $(a_i^+, a_i^-, b_i^+, b_i^-, c_i^+, c_i^-, d_i^+, d_i^-)$. 

We create one hospital $h_j$ corresponding to each clause $C_j$. 
The hospital $h_j$ has a capacity of $1$ and values each literal in the corresponding clause at $1$, all copies included. 
We also create two hospitals $s_i^+,s_i^-$ for each variable $x_i$ who we call the positive sink and the negative sink respectively. 
These two sink hospitals have a capacity of $4$ each. 
The positive sink $s_i^+$ values the doctors $x_i, x_i', a_i^+, b_i^+, c_i^+$ and $d_i^+$ at $1$; similarly, the negative sink $s_i^-$ values the doctors $\neg x_i, \neg x_i', a_i^-, b_i^-, c_i^-$ and $d_i^{-}$ at $1$. 
We also create a priority hospital $p_i$ for each variable $x_i$; the priority hospital for variable $x_i$ values the doctors $a_i^+, a_i^-, b_i^+, b_i^-, c_i^+,$ and $ c_i^-$ at $1$ and has a capacity of $3$.

Coming to the preferences of the doctors, the four positive dummy doctors of each variable $x_i$ ($a_i^+, b_i^+, c_i^+, d_i^+$) rank their priority hospital $p_i$ highest, then their positive sink $s_i^+$, and then all other hospitals. 
The four negative dummy doctors of each variable ($a_i^-, b_i^-, c_i^-, d_i^-$) rank their priority hospital $p_i$ highest, then their {\em negative} sink $s_i^-$, and then all other hospitals. 

The positive literals of each variable prefer their corresponding positive sink $s_i^+$ to any other hospital and the negative literals of each variable prefer their corresponding negative sink $s_i^-$ to any other hospital. 

Any ranking relation not mentioned in this construction can be filled up arbitrarily. 
This construction can be done in polynomial time, and its size is polynomial in the size of the original 2P2N-3SAT instance. 

The crucial property maintained by the construction is that for any stable allocation, if we restrict ourselves to the allocations to the clause hospitals $h_1,\dots,h_m$, a variable can be allocated in positive form or negative form {\em but not both}. 
This is because to allocate a positive (negative) literal of a variable $x_i$ to a clause hospital $h_j$, the positive (resp. negative) sink hospital of that variable must be filled up, i.e., the positive (resp. negative) sink must receive its capacity of four doctors to not create a blocking pair: otherwise, since the positive (resp. negative) sink hospital values the positive (resp. negative) literal doctor at $1$, they form a blocking pair.  
To fill up either sink, the priority hospital $p_i$ for that variable must be filled up; this is because all the doctors who are valued by the sink hospitals are also valued by the priority hospitals, whom these doctors strongly prefer. 
The total capacity of $s_i^+,s_i^-$ and $p_i$ is $11$ (the positive and negative sinks have a capacity of 4 and the priority hospital has a capacity of 3); however, they together value $12$ doctors at $1$ --- the four literal doctors $x_i,x_i',\neg x_i,\neg x_i'$ and the eight dummy doctors --- filling them up means there is only one doctor (either a positive or negative literal) corresponding to that variable that can be allocated to the clauses. 
Therefore, the allocation to the clauses cannot have two different kinds of literals for the same variable. 
This gives us the consistency needed to map stable allocations to SAT solutions. 

The rest of the proof is straightforward: if there is a solution to the 2P2N-3SAT instance, then there is a straightforward allocation where all hospitals get a utility of at least $1$. 
If a variable $x_i$ is positive in the original solution, we fill up the priority hospital with $a_i^-, b_i^-, c_i^-$, the positive sink with $a_i^+, b_i^+, c_i^+$ and $d_i^+$, and then allocate the doctors corresponding to the positive literals to the clauses which have these literals. We can do this for all the variables. If there are multiple literals which satisfy a clause, we can pick one arbitrarily. Note that this allocation gives all the clauses a positive utility. 
The negative and positive sinks have a positive utility in any stable allocation because of the doctors $d_i^+$ and $d_i^-$ who can only be allocated to these hospitals. 

Assume there is no solution to the 2P2N-3SAT instance. In any stable allocation, from our above discussion, the set of hospitals corresponding to clauses must receive either the positive form of a literal or the negative form of a literal but not both. This choice can be thought of as an assignment to the variables. Since the 2P2N-3SAT instance is unsatisfiable, there is at least one clause that is not satisfied by this assignment; this clause corresponds to a hospital which receives a utility of $0$ in the stable allocation. Since all assignments are unsatisfiable, all stable allocations give at least one hospital a utility of $0$.
\end{proof}

This result implies the intractability of other forms of individual fairness. For example, it implies that computing a stable egalitarian allocation, i.e., one that maximizes $\min_h v_h(X_h)$ subject to stability, is intractable.

\section{Conclusions and Future Work}
In this work we offer a comprehensive analysis of two-sided matchings where one side has MRF valuations over the other. We offer both sequential and optimization based approaches to computing stable, approximately efficient and approximately strategyproof allocations. The problem of computing stable, hospital strategyproof allocations remains open. Expanding these results to superclasses of MRF valuations is an important direction as well.

\section*{Acknowledgments}
Viswanathan and Zick are supported by an NSF grant RI-2327057. The work of Eden
was supported by the Israel Science Foundation (grant No. 533/23). Alon Eden is the Incumbent of the Harry \& Abe Sherman Senior Lectureship at the School of Computer Science and Engineering at the Hebrew University.

\bibliographystyle{plainnat}
\bibliography{abb,references}

\newpage
\appendix
\section{Missing Proofs from Section \ref{sec:preliminaries}}\label{apdx:prelims}

\obsblocking*
\begin{proof}
\begin{description}[leftmargin=0cm]
\item[$(\Longrightarrow)$]
If a blocking pair $(h, S)$ exists, then since $v_h(S) > v_h(X_h)$, there is some doctor $d \in S$ such that $\Delta_h(X_h, d) = 1$; this follows from the matroid augmentation property (\Cref{obs:augmentation}). 
For this doctor $d \in S$, from the definition of a blocking pair, $h \succ_d X(d)$. 
\item[$(\Longleftarrow)$] if there exists some doctor $d$ and hospital $h$ such that $\Delta_h(X_h, d) = 1$ and $h \succ_d X(d)$, then it is easy to see that $(h, X_h + d)$ forms a blocking pair: $v_h(X_h+d)> v_h(X_h)$, and for all $d' \in X_h+d$ $h \succeq_{d'} X(d')$. 
\end{description} 
\end{proof}

\lemgeneralwelfare*
\begin{proof}
    Since doctors have a strict and complete preference ordering over hospitals, for an allocation to be stable, it has to be maximal --- no unmatched doctor can have a non-zero marginal value for any hospital given the current assignment.  It is well known that when hospitals have submodular valuations, maximal allocations give a 2-approximation to the optimal hospital-USW~\cite{lehmann2006auctions}.
\end{proof}

\section{Missing Proofs for the HWSD Mechanism}\label{apdx:hwsd}

\begin{theorem}\label{thm:red-jacket-apdx}
The high welfare Serial Dictatorship (HWSD) mechanism has the following properties:
\begin{enumerate}[(i)]
    \item The mechanism is doctor-SP.
    \item The algorithm runs in polynomial time.
    \item The output allocation maximizes hospital-\USW.
    \item The output allocation is stable.
\end{enumerate}
\end{theorem}
\begin{proof}
We divide this proof into a series of claims. 
\begin{claim}
\Cref{algo:red-jacket} outputs an allocation maximizes hospital-\USW. 
\end{claim}
\begin{proof}
At every round, we maintain the invariant that there exists a socially optimal allocation $Y$ so that each hospital $h$'s assigned bundle of doctors $X_h$ is contained in $Y_h$. Thus, when the algorithm terminates, no additional doctors can be assigned, and the resulting allocation maximizes hospital-\USW.
\end{proof}

\begin{claim}
The mechanism is doctor-SP.
\end{claim}
\begin{proof}
The mechanism is doctor-SP since each doctor $d$'s preferences are only used to assign them to their highest ranked hospital that maintains non-redundancy and satisfies the max \USW invariant. Since the max \USW invariant and non-redundancy do not depend on the doctor $d$'s preferences, misreporting preferences can only lead to a worse assignment for $d$ (and any doctor by extension).
\end{proof}

\begin{claim}
\Cref{algo:red-jacket} computes a stable allocation.
\end{claim}
\begin{proof}
A new doctor only gets added to a hospital's allocation if it does not violate non-redundancy. Therefore, the final allocation must be non-redundant. 

Assume for contradiction that the output allocation (say $X$) has a blocking pair. By \Cref{obs:blocking}, there must be some doctor $d$ and some hospital $h$ such that $\Delta_h(X_h, d) = 1$ and $h \succ_d X(d)$. If $d$ is not allocated in $X$, then we get that $d$ can be added to $X_h$, and the social welfare strictly increases, contradicting the fact that the algorithm finds a hospital-\USW maximizing allocation. Therefore, assume that $h_d= X(d)$ is the hospital $d$ was assigned by \Cref{algo:red-jacket}. 
Suppose that we move the doctor $d$ from $X_{h_d}$ to $X_h$. 
This creates a new non-redundant allocation $X'$ where $X'_{h_d} = X_{h_d}-d$ and $X_{h}' = X_{h}+d$, and the remaining bundles are the same. Since the output of \Cref{algo:red-jacket} is non-redundant and $\Delta_{h}(X_{h},d)=1$, we have that $X'$ is also hospital-\USW maximizing: hospital $h_d$ lost a utility point and hospital $h$ gained one. 
    
Let $X^t$ be the interim allocation at the beginning of the $t$-th iteration of \Cref{algo:red-jacket}. Since doctors are never unassigned from hospitals, we have that for every hospital $h \in H$, $X_h^t \subseteq X_h^{t+1}$. 
Consider the iteration $t$ where $d$ was allocated and let $X^t$ be the allocation at the start of this iteration.  
Since $X_h^t$ is a subset of the final allocation $X_h$ to $h$ and $\Delta_h(X_h,d)= 1$, we must have $\Delta_h(X^t_h, d) = 1$ by submodularity. 
By construction there exists a max hospital-\USW allocation $X'$ with $X_{h'} \subseteq X'_{h'}$ for all $h' \ne h$, and $X_h^t + d \subseteq X'_h$. 
This implies that the hospital $h_d$ allocated by \Cref{algo:red-jacket} must satisfy $h_d \succ_d h$ contradicting the definition of a blocking pair. 
So, no such blocking pair exists, and the allocation is stable.
\end{proof}

\begin{claim}
\Cref{algo:red-jacket} runs in polynomial time.
\end{claim}
\begin{proof}
\Cref{algo:red-jacket} runs a \texttt{for} loop over the doctors and then examines hospitals in decreasing order of doctor prefernece in another \texttt{for} loop, for a maximum of $nm$ iterations. 
In each iteration, we test whether there exists some hospital-\USW maximizing allocation $Y$ such that $X_h \subseteq Y_h$ for all $h \in H$.   
Thus, to prove our claim we need to show that it is possible to compute the maximum hospital-\USW possible given a partial allocation $X$. 
Computing the maximum hospital-\USW allocation can be done in polynomial time using value queries, see e.g., \citet{benabbou2020finding}, who reduce the problem to the matroid intersection problem \cite{edmonds1979matroidintersection}. 

We also note that for any $T \subseteq D$, if $v_h$ is a matroid rank function, $v'_h$ defined as $v'_h(S) = v_h(T \cup S) - v_h(T)$ is also a matroid rank function. 
Therefore, we can compute the maximum hospital-\USW possible with hospital valuations $v'_h$ defined as $v'_h(S) = v_h(X_h \cup S) - v_h(X_h)$ for each $h \in H$. We can therefore efficiently determine if there exists a max hospital-\USW $Y$ such that for each $h$, $X_h \subseteq Y_h$. 
We conclude that \Cref{algo:red-jacket} runs in polynomial time.
\end{proof}
\end{proof}


\corredjacketadditive*
\begin{proof}
To prove this, we prove the following improved version of Lemma \ref{lem:manip-half}.

\begin{lemma}
Consider some hospital $h$. Fixing the reports of all other hospitals, let $X$ and $Y$ be the resulting allocation when $h$ reports $f_T$ and $v_h$, respectively. When agents report capped additive valuations, if $v_h(T)=|T|$ and $X_h=T$, then $|Y_h| \ge |T|$.
\end{lemma}
\begin{proof}
Similar to Lemma \ref{lem:manip-half}, we call the preference profile where $h$ reports $f_T$ the {\em misreported preferences}, and the profile where $h$ reports $f_T$ the {\em true preferences}.
Assume for contradiction that $|Y_h| < |T|$. Since the allocation $X$ is non-redundant with respect to the true preferences, it must be worse according to the HWSD ordering than $Y$. This implies $\USW(Y) \ge \USW(X)$. 

Therefore, since $|Y_h| < |T| = |X_h|$, there must be some hospital $h'$ such that $|Y_{h'}| > |X_{h'}|$. Let $d$ be some doctor in $Y_{h'} \setminus X_{h'}$. We create a path starting from the doctor $d$ just like we have done in the previous lemmas. However, since all agents have capped additive valuations, we can be less rigorous with how we define these paths. We move the doctor $d$ to the hospital $h_1$ it is allocated to in the allocation $X$. If this transfer results in $h_1$'s bundle exceeding its capacity, then there must be some doctor $d_1$ in $Y_{h_1}$ which is allocated to some other hospital (say $h_2$) in $X$. We repeat this process with $d_1$ till we make a transfer that respects the capacity constraints of each hospital. This creates a path $P = (h_0 = h', d_1 = d, h_1, d_2, \dots, d_k, h_{k+1})$. We define a transfer along this path as moving the doctor $d_i$ from hospital $h_{i-1}$ to $h_i$. The transfer along the path $P$ when applied to the allocation $Y$ gives us an allocation $Y'$ which is non-redundant with respect to both the true preferences. 

Since $|Y_h| < |X_h|$, the only place that hospital $h$ could appear in this path is right at the end (at $h_{k+1}$), since an addition of a single doctor to hospital $h$ (under the allocation $Y$) will not violate its capacity constraints. This means if we reverse the path and apply it to the allocation $X$, the new allocation (say $X'$) will be non-redundant with respect to the misreported preferences.  

Since $X$ and $X'$ are non-redundant with respect to the misreported preferences $X$ must be better than $X'$ according to the HWSD ordering (\Cref{obs:hwsd-ordering}). Since the only difference between the two allocations are the agents on the path, the allocation $X' \cap D_P$ is worse than the allocation $X \cap D_P$ according to the HWSD ordering. We now observe that $X' \cap D_P$ is equal to $Y \cap D_P$ and $X \cap D_P$ is equal to $Y' \cap D_P$. This implies $Y \cap D_P$ is worse than $Y' \cap D_P$ according to the HWSD ordering. Since the allocations $Y$ and $Y'$ only differ in the assignments of the doctors in $D_P$, this implies $Y$ is worse than $Y'$ is worse according to the HWSD ordering. This contradicts the fact that Algorithm \ref{algo:red-jacket} outputs $Y$ under the true preferences.
\end{proof}

Plugging this improvement into the proof of Theorem \ref{thm:red-jacket-strategyproofness} gives us the required result. 
\end{proof}

\section{Missing Proofs for the SD Algorithm}\label{apdx:sd}
\begin{theorem}\label{thm:serial-dictatorship-apdx}
The doctor serial dictatorship mechanism has the following properties:
\begin{enumerate}[(i)]
    \item The mechanism is doctor strategyproof
    \item The output allocation is stable
    \item The algorithm runs in polynomial time
    \item The output allocation is $2$-approximately max hospital-\USW.
\end{enumerate}
\end{theorem}
\begin{proof}
Again, we prove this using a series of claims. Note that (iv) follows from \Cref{lem:stable_welfare}. 

\begin{claim}
Algorithm \ref{algo:doctor-round-robin-revisited} runs in polynomial time.  
\end{claim}
\begin{proof}
\Cref{algo:doctor-round-robin-revisited} runs in polynomial time since we simply go down the doctors' preference lists in decreasing order, and run a poly-time procedure to check whether each hospital is willing to accept the doctor. 
\end{proof}

\begin{claim}
\Cref{algo:doctor-round-robin-revisited} outputs a stable allocation.   
\end{claim}
\begin{proof}
A new doctor $d$ only gets added to a bundle $X_h$ if $\Delta_h(X_h,d) = 1$. 
Therefore the final allocation is non-redundant. 

If there is a blocking pair, there must be some $d$ and $h$ such that $\Delta_h(X_h, d) = 1$ and $h \succ_d X(d)$. Assuming for contradiction that such a pair exists. 
Consider the iteration where $d$ was allocated and let $Y$ be the allocation at the start of this iteration. By submodularity, we must have $\Delta_h(Y_h, d) = 1$ as well which implies that $X(d) \succ_d h$ contradicting the definition of a blocking pair. So, no such blocking pair exists, and the allocation is stable.
\end{proof}

\begin{claim}
\Cref{algo:doctor-round-robin-revisited} is doctor-SP.
\end{claim}
\begin{proof}
The mechanism is  doctor-SP since each doctor $d$ has a fixed iteration where they are assigned, and at this iteration they are assigned to their most preferred hospital who can accept them. 
As the hospitals' allocations at the iteration where the doctor $d$ is allocated is independent of the doctor's preferences, the doctor $d$ cannot obtain a strictly better hospital by misreporting their preferences.    
\end{proof}

\end{proof}


\section{Missing Proofs from Section \ref{sec:cardinal-utilities}}\label{apdx:cardinal-utilities}

\thmmaxnashwelfare*
\begin{proof}
We first show that the algorithm terminates in polynomial time. To do this, we show that the \textbf{while} loop terminates in polynomial time. This is sufficient since we run the \textbf{while} loop for $m$ different values of $k$. 

We use a potential function argument to bound the number of \textbf{while} loop iterations. Every time we change the allocation, we strictly reduce $\sum_{i \in H} \left (|X_i| + \frac{i}{n^2}\right)^2$. Therefore, the while loop can only run $\cal O(m^2n^2)$ times. We note that this exact potential function argument has been used before in \citet{babaioff2021EF} and \citet{cousins2023mixedmanna}

Next, we prove the correctness of our algorithm.
Each allocation $X^k$ computed by the algorithm maximizes doctor-\USW subject to the constraint that total hospital-\USW is $k$. Since we pick an allocation $X^k$ that maximizes doctor-\USW, the allocation $X$ output by \Cref{algo:max-nash-welfare} maximizes doctor-\USW.
Next, we show that $X$ maximizes hospital Nash welfare. 
Let $Y$ be an allocation that 
\begin{inparaenum}[(a)]
    \item maximizes doctor-\USW, 
    \item hospital-\NSW  subject to (a), and 
    \item lexicographically dominates all other allocations that satisfy (a) and (b). 
\end{inparaenum}
Furthermore, assume that $Y$ has a hospital social welfare of $k$. 

Let us consider the allocation $X^k$ that Algorithm \ref{algo:max-nash-welfare} computes. 
Assume for contradiction that $X^k \ne Y$. Since $X^k \ne Y$, there is some hospital $i$ such that $|X_i^k| < |Y_i|$. If there are multiple such hospitals, we break ties by choosing the $i$ with lower $|X_i|$, and breaking further ties by choosing the lower $i$. 
Since both $X^k$ and $Y$ have a total hospital-\USW of $k$, there is some hospital $j$ such that $|Y_j| < |X_j^k|$; again, if there are multiple such hospitals, we break ties by choosing the $j$ with lower $|Y_j|$, and breaking further ties by choosing the lower $j$. . 

We have four possible cases, each leading to a contradiction proving that $X^k = Y$.

\begin{description}[leftmargin=0cm]
    \item[Case 1: $|X_i^k| = |Y_j|$ and $i <j$.] We apply \Cref{lem:max-nash-welfare-main} with the allocation $X^k$ and hospital $i$ to find a hospital $\ell$ and a $k$-MDW allocation $Z$ such that the hospital $i$ gets one more doctor, the hospital $\ell$ gets one less doctor and all other hospitals have the same bundle size. We have $|X^k_{\ell}| - 1 \ge |Y_{\ell}| \ge |Y_j|  = |X^k_i|$. 
    If any of these inequalities is strict, then $|X^{k}_{\ell}| \ge |X^k_i| + 2$; in this case, \Cref{algo:max-nash-welfare} reduces the capacity of $\ell$ by 1 and increases the capacity of $i$ by $1$, and moves from the allocation $X^k$ to $Z$, which contradicts $X^k$ being the resulting allocation. 
    If equality holds throughout, then $|X^k_{\ell}| = |X^k_i| + 1$ and $\ell >j > i$, which again contradicts the logic of \Cref{algo:max-nash-welfare}.
    \item[Case 2: $|X^k_i| < |Y_j|$.] This case follows a similar argument to that of Case 1. 
    \item[Case 3: $|X^k_i| = |Y_j|$ and $j < i$.] In this case, we plug in Lemma \ref{lem:max-nash-welfare-main} with allocation $Y$ and hospital $j$ to find a hospital $\ell$ and a $k$-MDW allocation $Z$ such that $j$ gets one more doctor, $\ell$ gets one less doctor and all other hospitals have the same bundle size. We have $|Y_{\ell}| - 1 \ge |X^k_{\ell}| \ge |X^k_i|  = |Y_j|$. 
    If any of these inequalities is strict, then $|Y_{\ell}| \ge |Y_j| + 2$ which means $Z$ has a higher hospital Nash welfare than $Y$, a contradiction to $Y$ being hospital-\NSW optimal. 
    If equality holds throughout, then $|Y_{\ell}| = |Y_j| + 1$ and $\ell > i> j$, which implies that $Z$ is hospital-\NSW optimal as well but lexicographically dominates $Y$, which again violates our assumption on $Y$.
    \item[Case 4: $|X^k_i| > |Y_j|$.] The proof of this case follows a similar argument to that of Case 3. 
\end{description}
Therefore, $X^k = Y$, and the Nash welfare of our output allocation $X$ is at least that of $Y$.
\end{proof}

\end{document}